\documentclass[draft]{article} 
\usepackage{fullpage} 
\usepackage{multicol} 
\usepackage{mathptmx} 
\usepackage{amsmath}
\usepackage{amssymb}
\usepackage{amsthm} 
\newtheorem{theorem}{Theorem}
\newtheorem{definition}[theorem]{Definition}
\newtheorem{proposition}[theorem]{Proposition}
\newtheorem{corollary}[theorem]{Corollary}

\usepackage[pdftex]{graphicx}
\usepackage{LatexPackages/amssymb} 
\usepackage{LatexPackages/csquotes} 
\usepackage{LatexPackages/tikz-proofnets} 
\usetikzlibrary{positioning} %% Sylvain: allows for using right= 2em of
\usetikzlibrary{decorations.markings} %% Sylvain: allows for arrow tips not at the end of the arrow
\tikzset{vertex/.append style={inner sep=1.5pt}} %% Sylvain: enlarge the circles for vertices
\usepackage{subfigure}

\usepackage{tabulary} %% Sylvain: to get line width tables.
\usepackage{LatexPackages/prooftree} 

\usepackage{skull} 
\usepackage{url} 

\newcommand\un{\mathbf{1}} 

\renewcommand\sp{\ifmmode{\scriptscriptstyle SP}\else{\sc sp}\fi} 

\newcommand\ael{{\sc ae}}

\newcommand\edge{\!\!\--\!\!} 
 
\newcommand\fle{\!\rightarrow\!} 

\newcommand\aretearc[2]{\stackrel{#1}{#2}}
\newcommand\arete[1]{\aretearc{#1}{\longleftrightarrow}} 
\newcommand\aretevers[1]{\aretearc{#1}{\longrightarrow}} 
\newcommand\aretede[1]{\aretearc{#1}{\longleftarrow}}

\usepackage{etoolbox}
\apptocmd{\thebibliography}{\renewcommand{\sc}{}}{}{}

\usepackage{LatexPackages/stmaryrd} 

\newcommand\restr\restriction 

\newcommand\dicogpn{{\sc dicog-PN}} 
\newcommand\dicogr{{\sc dicog-RS}}

\newcommand\flag{\,^<\!\!|} 
\newcommand\pa{\bindnasrepma} 
\newcommand\pasp{\widehat{\bindnasrepma}} 
\newcommand\bef{\mathbin{\triangleleft}} 

\newcommand\befsp{\mathbin{\widehat{\triangleleft}}} 
\renewcommand\precedesOP{\bef} %% Sylvain: pour que dans les liens ce soit le même symbole
\renewcommand\llprecedes{\bef} %% Sylvain: pour que dans les liens ce soit le même symbole
\newcommand\lts\ts 
\newcommand\ts{\mathbin{\otimes}} 
\newcommand\tssp{\mathbin{\widehat{\otimes}}} 

\newcommand\simpa{\stackrel{\cdot\,\cdot}{\sim}}
\newcommand\simbef{\stackrel{\rightarrow}{\sim}}
\newcommand\simts{\stackrel{\--}{\sim}}

\newcommand\tablecoh[5]
{
\begin{array}{l||c|c|c|}
#1 #3 #2 & \sincoh & = & \scoh \\ \hline \hline 
\sincoh & \sincoh & \sincoh & #4 \\ \hline 
= & \sincoh & = & \scoh \\ \hline 
\scoh & #5  & \scoh & \scoh \\ \hline 
\end{array} 
}

\newtheorem{criterion}{\textbf{Correctness criterion}}
\newtheorem{interrogation}{\textbf{Question}}

\DeclareMathAccent{\arc}{\mathord}{letters}{"7E}

\newcommand\seq\vdash 

\usepackage{makeidx}         % allows index generation
\usepackage{graphicx}        % standard LaTeX graphics tool
\makeindex             % used for the subject index
\newcommand\vinf{\mathit{vinf}}
\newcommand\np{\mathit{np}}
\newcommand\NEG{\mathit{neg}}
\newcommand\obj{\mathit{obj}}
\newcommand\NE{\mathit{ne}}
\newcommand\pas{\mathit{pas}}

\title
{Pomset Logic\\ the other approach to non commutativity in logic}
\author{Christian Retor\'e
--- LIRMM Univ Montpellier \& CNRS} 
\date{\footnotesize Near-final version of Retoré, C. (2021). Pomset Logic. In: Casadio, C., Scott, P.J. (eds) Joachim Lambek: The Interplay of Mathematics, Logic, and Linguistics. Outstanding Contributions to Logic, vol 20. Springer, Cham.  
\url{https://doi.org/10.1007/978-3-030-66545-6_9}} 

\begin{document}


%\institute{Christian Retor\'e \at LIRMM, Univ Montpellier, CNRS,  Montpellier, France\newline \email{christian.retore@umontpellier.fr}

%\and Name of Second Author \at Name, Address of Institute \email{name@email.address}

\maketitle
\begin{abstract}Thirty years ago, I introduced a non-commutative variant of classical linear logic, called \emph{pomset logic}, issued from a particular categorical interpretation of linear logic known as  coherence spaces. In addition to the usual commutative multiplicative connectives of linear logic, pomset logic includes a non-commutative connective, "$\bef$" called \emph{before},  associative and self-dual: $(A\bef B)^\perp=A^\perp \bef  B^\perp$. The conclusion of a pomset logic proof is a Partially Ordered Multi\textsc{set}  of formulas.  
Pomset logic enjoys a proof net calculus with cut-elimination, denotational semantics, and  faithfully  embeds sequent calculus. 
\newline\indent 
The study of pomset logic has reopened with recent results on handsome proof nets, on its sequent calculus,  or on its following calculi like deep inference by Guglielmi and Stra{\ss}burger. Therefore, it is high time we published a thorough presentation of pomset logic, including  published and unpublished material, old and new results. 
\newline\indent 
Pomset logic (1993) is a non-commutative variant of linear logic (1987) 
as for Lambek calculus (1958!) and it can also be used as a grammatical formalism.     
Those two calculi are quite different, but we hope that the algebraic presentation we give here, with formulas as algebraic terms and with a semantic notion of proof (net) correctness, better matches Lambek's view of what a logic should be.
\end{abstract} 

\vfill 

\begin{multicols}{2} 
\small 
\tableofcontents 
\end{multicols} 

\vfill 

\normalsize 
\newpage 
%\small 
%\tableofcontents
%\normalsize 
%\newpage 

\section{Presentation}
\label{pres}

Lambek use to refer to his \emph{logic} \cite{Lam58} with the words \emph{syntactic calculus} thus expressing his preference for algebra, thereafter confirmed with his move from categorial grammars to pregroup grammars,  which are not a logical system. 
Up to the invention of linear logic in the  late 80s,  Lambek calculus was a rather isolated logical system, despite some study of frame semantics, which are typical of substructural logics.

Linear logic \cite{Gir87} arose from the study of the denotational semantics of system F, itself arising from the study of ordinals \cite{Gir86}. 
For interpreting systems F (second order lambda calculus) with variable types, one needed to refine the categorical interpretation of simply typed lambda calculus with Cartesian Closed Categories. In order to quantify over types Girard considered the category of coherence spaces,  initially called qualitative domains,  with stable maps, which preserve directed joins and pullbacks. 
A finer study of coherence spaces led Girard to discompose the arrow type construction in to two steps: one is to contract several object of type $A$ into one (modality/exponential !) and the other one being linear implication (noted $\multimap$) which rather corresponds to a change of state than to a consequence relation. 

Linear logic was first viewed as a proof system (sequent calculus or proof nets) which is well interpreted by coherence spaces. The initial article \cite{Gir87} also included the definition of phase semantics, that resembles frame semantics developed for the Lambek calculus. 
It was not long before the connection between linear logic and Lambek calculus was found: after some early remarks by Girard, Yetter \cite{Yet90} observed the connection at the semantic level, while Abrusci \cite{Abru91} explored the syntactic, proof theoretical connection, while \cite{Ret96tal} explored proof nets and completed the insight of \cite{Roo92}. Basically Lambek calculus is non commutative intuitionistic multiplicative logic, the order between the two restrictions, intuitionistic and non commutative, being independent. An important remark, that I discussed with Lamarche in \cite{LR96}, says that non commutativity requires linearity in order to get a proper logical calculus.

Around 1988,  my PhD advisor Jean-Yves Girard pointed to my attention  a binary non commutative connective $\bef$ in coherence spaces. In coherence spaces, this connective has intriguing properties: 
\begin{itemize} 
\item $\bef$ is self dual $(A\bef B)^\perp\equiv (A^\perp \bef B^\perp)$, without swapping the two components   --- by $X \equiv Y$ we mean that there is a    pair of canonical invertible linear maps between $X$ and $Y$. 
\item $\bef$  is non commutative  $(A\bef B) \not\equiv (B\bef  A)$
\item $\bef$ is associative $((A\bef B)\bef C)\equiv (A\bef (B\bef C))$; 
\item it lies  in between the commutative conjunction $\ts$ and disjunction $\pa$ 
there is a canonical linear map   from $(A\ts B)$ to $(A\bef B)$ an one from $(A\bef B)$ to $(A\pa B)$  --- remember that coherence spaces validate the \textsc{mix} rule $(A\ts B)\multimap (A\pa B)$; 
\end{itemize} 
I designed a proof net calculus with this connective, in which a sequent, that is the conclusion of a proof,  is a partially ordered multiset of formulas. This proof net calculus  enjoys cut-elimination and a sound and faithful
(coherence) semantics, in the sense that having an interpretation (which is preserved under cut elimination) is the same as being syntactically correct cf. Theorem \ref{semcorrect}. 
I proposed a  version of sequent calculus that easily translates into those proof nets and enjoys cut-elimination as well \cite{Ret93}.  However despite many attempts by me and others (Sylvain Pogodalla, Lutz Stra{\ss}burger) over many years we did not find a sequent calculus that would be complete w.r.t. the proof nets. Later on, Alessio Guglielmi, soon joined by Lutz Stra{\ss}burger, designed the calculus of structures, a term calculus more flexible than sequent calculus (deep inference)  with the \emph{before} connective \cite{Gug94,Gug99,GugStr01}, a system that is quite close to dicograph rewriting \cite{Ret98roma,Ret99romarr}  They tried to see whether one of their systems called SBV was equivalent to pomset logic; here we include in Section \ref{sbv} a new result saying deep pomset  i.e. a rewrite view of pomset logic proof nets corresponds to  SBV as well as  an old result of us saying that SBV tautologies are  provable formulas of pomset proof nets. 

As a reviewer of my habilitation \cite{RetoreHDR} Lambek wrote: 
\begin{quotation} \it 
He constructs a model of linear logic using graphs, which is new to me. His most original contribution is probably the new binary connective which he has added to his non commutative version of linear logic, although I did not find where it is treated in the sequent calculus. (J. Lambek, Dec. 3, 2001) 
\end{quotation} 
I deliberately omitted my work on sequent calculus in my habilitation manuscript, because  none of  the sequent calculi I experimented with  was complete w.r.t. pomset proof nets which are "perfect", i.e. enjoy all the expected proof theoretical properties. In addition, by that time, I did not yet have a counter example to my proposal of a sequent calculus, the one in Figure  \ref{counterex} of Section \ref{sequentialisation} was found ten years later with Lutz Stra{\ss}burger. 

However, very recently, Slavnov found a sequent calculus that is complete w.r.t. pomset proof nets \cite{Slavnov2019}.  The structure of the decorated sequents that Slavnov uses is rather complex\footnote{A decorated sequent according to Slavnov is  a multiset of pomset formulas $A_1,\ldots,A_n$ with 
$p\leq n/2$ binary relations $(R_k)$ for $1\leq k\leq p$ between sequences of length $p\leq n/2$ of formulas  from $\Gamma$; those relations are such that whenever $(B_1,\ldots,B_k R_k (C_1,\ldots,C_k)$ the two sequences $(B_1,\ldots,B_k)$ and $(C_1,\ldots,C_k)$ have no common elements and 
$(B_1,\ldots,B_k) R_k (C_1,\ldots,C_k)$ entails $(B_{\sigma(1)},\ldots,B_{\sigma(k)}) R_k (C_{\sigma(1)},\ldots,C_{\sigma(k)})$ for any permutation $\sigma$ of $\{1,\ldots,k\}$ -- those relations correspond to the existence of disjoint paths in the proof nets from $B_i$ to $C_i$.}
and the connective $\bef$ is viewed as the identification of two dual connectives one being more like a $\ts$ and the other more like a $\pa$. As this work is not mine I shall not say much about it, but Slavnov's work really  sheds new light on pomset logic.  Given the complexity of Slavnov sequent calculus,  it is enjoyable to provide a simple sequent calculus for pomset logic, even though it does not generate all proof nets of pomset logic. 

Pomset logic and the Lambek calculus systems share some properties: 
\begin{itemize} 
\item They both are linear calculi; 
\item They both handle non commutative connective(s) and structured sequents; 
\item They both have a sequent calculus; 
\item They both enjoy cut-elimination; 
\item They both have a complete sequent calculus (regarding  pomset logic the complete sequent calculus is quite new);  
\item They both can be used as a grammatical system. 
\end{itemize} 

However Lambek calculus and pomset logic are quite different in many respects: 
\begin{itemize} 
\item Lambek calculus is naturally an intuitionistic calculus while pomset logic is naturally a classical calculus --- although in both cases variants of the other kind can be defined, cyclic linear logic of Yetter is a classical version of Lambek calculus \cite{Yet90}, and the LLMS system of Reddy is a term calculus for the semantics of higher-order imperative languages, and   it can be considered as an intuitionistic  version of pomset logic \cite{Red93}.\footnote{This model of higher-order imperative programming with "before" strongly inspired the author subsequent model \cite{Red96}, although it is not really a logical system anymore.}  
\item Lambek calculus is a restriction of the usual multiplicative linear logic  according to which the connectives are no longer commutative, while pomset logic is an extension of usual commutative multiplicative linear logic with a non commutative connective. 
\item Lambek calculus deals with totally ordered multisets of hypotheses while pomset logic deals with partially ordered multisets of formulas. As grammatical systems, pomset logic allows relatively free word order, while Lambek calculus only deals with linear word orders. 
\item Lambek calculus has an elegant truth-value interpretation within the subsets of a monoid (frame semantics, phase semantics), while there is not such a notion for pomset logic. 
\item Lambek calculus has no simple concrete  interpretation of proofs up to cut elimination  (denotational semantics) while coherence semantics faithfully interprets the proofs of pomset logic. 
\end{itemize}

This list shows that those two comparable systems also have many differences. However, the presentation of Pomset logic provided by the present article make Lambek calculus and pomset logic rather close on an abstract level. As he often said, Lambek did not like standard graphical or geometrical presentation of linear logic like proof nets. He told me several times that moving from geometry to algebra has been a great progress in mathematics and solved many issues, notably in geometry, and that proof net study was going the other way round.  I guess this is related to what he said about Theorem \ref{semcorrect} in the present paper:  

\begin{quotation} \it 
It seems that this ingenious argument avoids the complicated long trip condition of Girard. It constitutes a significant original contribution to the subject. (J. Lambek,Dec 3 2001) 
\end{quotation}

This paper is a mix (!) of easy to access published work, \cite{FR94,Ret96tal,BGR97,Ret97tlca,Ret97,Ret96entcs,Ret03tcs}  research reports and more confidential publications \cite{Ret93,Ret93b,Ret93b,LR95,Ret94rc,LeRe96,LR96,LeRe97,Ret98roma,Ret99romarr,Ret99rr,RetoreHDR,PogoReto04}, unpublished material between 1990 and 2020, that are all  presented in the same and rather new unified perspective; the presented material can be divided into four topics: 
\begin{description}
\item[\bf proof nets] 
handsome proof nets both for MLL Lambek calculus and pomset logic, and other work on proof nets \cite{FR94,LR96,Ret96tal}, 
\item[\bf combinatorics] 
(di)cographs and \sp\ orders  \cite{Ret93b,Ret93,Ret96entcs,BGR97,Ret98roma,Ret99romarr,Ret99rr,RetoreHDR,Ret03tcs,RetoreHDR,Ret03tcs,PogoReto04}, 
\item[\bf coherence semantics]  \cite{Ret93,Ret94rc,Ret97,RetoreHDR}, 
\item[\bf grammatical applications]  of pomset logic to computational linguistics \cite{LR95,LeRe96,LeRe97,Ret98roma,Ret99romarr,RetoreHDR}.
\end{description}
The contents of the present article is divided into six sections as follows: 
\begin{enumerate}
 \setcounter{enumi}{1}
 \item As a starter, we offer a glimpse of pomset logic, an informal tour which summarizes the most important constructions and results of pomset logic. 
     \item  We then present results on series parallel partial orders, cographs and dicographs that subsumes those two notions and present dicograph either as \sp\ pomset of formulas or as dicographs of atoms, and explain the guidelines for finding a sequent calculus. This combinatorial part is a prerequisite for the subsequent sections. 
    \item Proof nets without links, the so called handsome proof nets, are presented as well as the cut elimination for them. 
    \item The semantics of proof nets, preserved under cut elimination and equivalent to their syntactic correctness is then presented. 
    \item Then the sequentialisation "the quest" of a complete sequent calculus is discussed and we provide an example of a proof net that does not derive from any simple sequent calculus. We do not present Slavnov sequent calculus which is quite complicated and thoroughly presented in his recent paper \cite{Slavnov2019}. 
 \item We then present proofs in an algebraic manner, à la depp inference, with deduction rules as term rewriting and show the correspondence btween this view and the calculus of structures known as SBV. 
    \item Finally we explain how one can design grammars by associating words with  partial proof nets of pomset logic. 
\end{enumerate}

\section{A glimpse of pomset logic} 
\label{glimpse}

When asked for a presentation of a different sort of logic, the preferred way of most readers is to provide them with a sequent calculus. Hence we shall give a simple sequent calculus which is a subcalculus of the sequents that pomset logic is able to derive.

The formulas of pomset logic are defined from atoms (propositional variables or their negation) by means of the usual commutative multiplicative connectives $\pa$ and $\ts$ together with the new non commutative multiplicative  connective $\bef$ (before)--- the three of them are associative.  

\newcommand\FFF{\mathcal{F}}
\newcommand\PPP{\mathcal{P}} 

$$\FFF\ ::=\ \PPP\ |\ \PPP^\perp\ |\ \FFF\ts \FFF\ |\ \FFF\pa \FFF\ |\ \FFF\bef \FFF$$

It is assumed that formulas are always in negative normal form: negation only apply to propositional variables; this is possible and  standard when negation is involutive and satisfies the De Morgan laws:

$$
\begin{array}{rcl} 
(A^\perp)^\perp&=&A\\ 
(A\pa B)^\perp&=&(A^\perp\ts B^\perp)\\ 
(A\bef B)^\perp&=&(A^\perp\bef B^\perp)\\ 
(A\ts B)^\perp&=&(A^\perp\pa B^\perp)\\ 
\end{array}
$$

Sequents of pomset logic are right handed and they are partially ordered multisets of formulas (pomsets of formulas). We assume those partial orders are described by operation from the one point order. Although we shall be much more precise in the next section (Section \ref{dicographs}) about partial orders, we need to define two operations on partial orders, at least informally. Given two partially ordered multisets of formulas, $\Gamma$ and $\Delta$, let us define two orders whose domain is the disjoint unions of the two domains and which preserve order on each domain: 
\begin{itemize}
    \item 
$\{\Gamma,\Delta\}$ their parallel composition: any two formulas one of them in  $\Gamma$ and the other one in $\Delta$ cannot be compared. This operation is associative and commutative. 
\item 
$\langle \Gamma ; \Delta \rangle$ their series composition: any formula in $\Gamma$ is smaller than any formula in $\Delta$. This operation is associative, but non commutative. 
\end{itemize}

The expression $\Gamma[X]$
denotes any pomset including a propositional variable $X$, 
and given a pomset $\Delta$ the expression $\Gamma[\Delta]$ denotes the pomset obtained by substituting in  $\Gamma[X]$ the formula $X$ with the pomset $\Delta$  as a term.

\begin{figure} 
\begin{center} 
\input{pomset_orders_intuitive}
\end{center} 
\caption{A simple sequent calculus for pomset logic} 
\label{simpleseqcalc} 
\end{figure} 

\begin{figure} 
$$
\begin{prooftree}
\[ 
\[ %1
\[
\seq \{a, a^\perp\} 
\justifies 
\seq a \pa a^\perp  
\] 
\qquad 
\[
\seq \{b,b^\perp\} 
\justifies 
\seq b\pa b^\perp  
\] 
\justifies 
\seq (a \pa a^\perp)\ts (b\pa b^\perp) 
\]%1
\qquad \qquad 
\seq c,c^\perp 
\justifies 
\seq \langle (a \pa a^\perp)\ts (b\pa b^\perp) ; \{c,c^\perp\} \rangle 
\using dimix 
\]
\justifies 
\seq \{\langle (a \pa a^\perp)\ts (b\pa b^\perp) ; c \rangle,  c^\perp\}  
\using entropy 
\end{prooftree} 
$$ 
\caption{Example a proof in pomset logic in the simple sequent calculus of Figure \ref{simpleseqcalc}.} 
\label{asequentproof} 
\end{figure}

This sequent calculus extends classical multiplicative linear logic. Orders can be "weakened" until the discrete order is reached. When $dimix$ is not used (hence $entropy$ cannot be used either) this calculus is MLL.

As sequent calculus is best suited for classical logic,  as intuitionistic logic fits in well with natural deduction, multiplicative linear logic is better expressed with proof nets, and this is even more striking in the pomset logic case. Nevertheless, for pedagogical reason 
% modification 27 sept 
we give
% fin de modification 27 sept 
a simple sequent calculus for pomset logic, which does not encompass all proof nets to be later defined. 

There is an elegant proof net calculus where to map the sequent calculus proofs, defined in Section \ref{PNling}  identifying the sequent calculus proofs that are essentially similar, like the ones obtained one from the other by commuting rules. In addition to the par and time links, one needs a link for before. Although a Danos Regnier criterion is absolutely possible, it is unnatural for this calculus, for which it is easier to use edge bicoloured graphs (blue and red) with undirected B edges and R edges, some of them being being directed. Here are the links: 

\begin{figure} 
\begin{center} 
\begin{tabulary}{\linewidth}{|c|C|C|C|C|C|}\hline 
 & Axiom & Par $\pa$ & Before $\bef$ & Times $\ts$ & Cut \\ \hline  
Premisses & None & $A$ and $B$ & $A$ and $B$ &$A$ and $B$ & $K$ and $K^\perp$ \\ \hline 
RnB link &  
 \begin{tikzpicture}[PS]  
    \node[vertex,label=-90:$a$] (a) {};
    \node[vertex,label=-90:$\llneg{a}$, right=2em of a] (acomp) {};
      \draw[axiom, controls=+(90:0.5) and +(90:0.5)] (a) to (acomp);
 \end{tikzpicture}
&
\begin{forest} proof structure, commutative, for tree={s sep=2.5em}
  [,label=0:$A\pa B$
  [,par
  [,label=180:$A$,name=leftA]
  [,label=0:$B$,name=rightB]
  ]
  ]
\end{forest} 
& 
 \begin{forest} proof structure, commutative, for tree={s sep=2.5em}
  [,label=0:$A\bef B$
  [,precedes
  [,label=180:$A$,name=leftA,]
  [,label=0:$B$,name=rightB]
  ]
  ]
\end{forest}
&
\begin{forest} proof structure, commutative, for tree={s sep=2.5em}
  [,label=0:$A\otimes B$
  [,times
  [,label=180:$A$,name=leftA,]
  [,label=0:$B$,name=rightB]
  ]
  ]
\end{forest}
&
 \begin{forest} proof structure, commutative, for tree={s sep=2.5em}
  [,fill=black
  [,cut %% Sylvain: replacing "times" -> do you want the times symbol inside?
  [,label=180:$A$,name=leftA,]
  [,label=0:$A^\perp$,name=rightB]
  ]
  ]
\end{forest}
\\ \hline 
Conclusion(s) & $a$ and $a^\perp$ & $A\pa B$ & $A\bef B$ & $A\ts B$ & None \\ \hline 
\end{tabulary} 
\end{center} 
\caption{The links of pomset logic as edge bicoloured graphs. In a proof structure,  the conclusion of a link is the premisse of at mots one link, and each premisse of a link is the conclusion of exactly one link. A formula that is not the premisse of any link is said to be a conclusion of the proof structure. Cuts are conclusions  $K\ts K^\perp$, they never can be the premisse of any link.} 
\label{rnblinks} 
\end{figure}

Proof nets are defined as the simple graphs defined from those links for which blue edges  define a perfect matching and without elementary circuits (directed cycles without twice the same vertex) alternating the B (axioms and formulas) and the R edges (connectives).

\begin{figure} 
\begin{center} 
\begin{forest} proof structure, with phantom node, commutative
[,phantom
  [,fit=band
  [,precedes
  [
  [,times
  [
   [,par  [,label=180:$a$,name=a] [,label=0:$\llneg{a}$,name=acomp] ]
    ]
  [
   [,par  [,label=180:$b$,name=b] [,label=0:$\llneg{b}$,name=bcomp]]
   ]
  ]
  ] 
  [,label=0:$c$,name=c]
  ]
  ]
    [,l*=3,label=0:$\llneg{c}$,name=ccomp]
  ]
\draw[axiom, controls=+(90:0.5) and +(90:0.5)] (a) to (acomp);
\draw[axiom, controls=+(90:0.5) and +(90:0.5)] (b) to (bcomp);
\draw[axiom, controls=+(90:0.5) and +(90:0.5)] (c) to (ccomp);
\end{forest}
\end{center}
\label{aproofnetwlinks} 
\caption{The proof net corresponding to the sequent calculus proof in Figure \ref{asequentproof}} 
\end{figure}

However, there is a much more interesting view of multiplicative proof nets, the so-called \emph{handsome proof nets} that I first introduced for usual multiplicative linear logic \cite{Ret96entcs,Ret99rr} do not have  links, as we shall see in Section \ref{handsomepn}. A handsome proof net  a graph which does not depend from the associativity and commutativity of the connectives. The logical formula is the R graph, the axioms linking atoms are the B edges  and the  criterion is: every alternating elementary circuit contains a chord (that is an edge directed or not linking two points on the circuit but not itself in the circuit).\footnote{In a proof net with links there cannot be chords on alternate elementary cycles. hence  this criterion when applied to proof net with links is the one we defined above for proof nets with links.}

\begin{figure} 
\begin{center} 
\begin{tikzpicture}[PS,scale=1.5]
\tikzset{vertex/.append style={fill=black}}
  %% La commande \vertex construit un sommet avec le label donné entre
  %% crochet (n'accepte pas de commande} et lui donne ce nom.
  %% Si en plus option 'neg', on rajoute un \perp au label et un
  %% 'comp' au nom du sommet
  %% Si on veut donner un autre label, rajouter l'option (on est en
  %% mode math par défaut, donc ne pas rajouter les $
  \vertex[orientation=180,on circle]{a};
  \vertex[orientation=120,neg,on circle]{a};
  \vertex[orientation=60,on circle]{b};
  \vertex[orientation=0,neg,on circle]{b};
  \vertex[orientation=-60,on circle]{c};
  \vertex[orientation=-120,neg,on circle]{c};

  \draw[Blink] (a) to (acomp);
  \draw[Blink] (b) to (bcomp);
  \draw[Blink] (c) to (ccomp);

  %% \rlink construit un lien rouge orienté
  \llop{\cutlink}{a,acomp}{b,bcomp}
  \llop{\rlink}{a,acomp,b,bcomp}{c}
%% Pour mettre la flèche pas tout au bout 
%%  \llop{\rlink[-,decoration={markings,mark=at position .95 with {\arrow[>=stealth]{>}}},postaction={decorate}]}{a,acomp,b,bcomp}{c}

  %% \cutlink construit un lien rouge non orienté

\end{tikzpicture}
\end{center} 
\label{ahandsomeproofnet} 
\caption{The handsome proof net corresponding to the proof net in Figure \ref{aproofnetwlinks} and to the sequent calculus proof in Figure \ref{asequentproof}} 
\end{figure}

This proof net calculus enjoys cut-elimination. Further more, cuts take part to the order on the conclusions, which might be view as the encoding of a strategy to reduce them, see Section \ref{handsomecuts}. 

Some graph rewriting rules preserve the correctness of handsome proof nets so we develop in section \ref{sbv} a notion of derivation of proof nets from axioms 
$(a_1\pa a_1^\perp) \ts \cdots \ts (a_n\pa a_n^\perp)$
which are themselves proof nets: this kind of derivation is called deep pomset logic; of course as shown in \cite{Ret98roma,Ret99romarr} rewriting preserves correctness. Such a rewrite view of pomset logic was just suggested at the end of \cite{Ret99rr}, but it was later developed by Gugliemi and Stra{\ss}burger in \cite{GugStr01} within the calculus of structures and deep inference. We here prove that their SBV calculus corresponds to deep pomset logic. Of course we would like to know whether the rewriting from axioms produces all correct proof nets, but we do not know. 

Pomset logic  is easily interpreted in coherence spaces. Proofs (proof nets) are interpreted as elements of the coherence space associated with the conclusion, in such away that this interpretation of proofs is preserved by cut-elimination,  see Section \ref{cohInt}. 

Furthermore, the fact that the interpretation of a proof net as a set of "experiments" is a semantic object (a clique of the corresponding coherence space) is equivalent to the correctness of the proof net,  see Section \ref{cohInt}.

\section{Structured sequents as dicographs of formulas}
\label{orders} 

In order to draw a distinction between $A\bef B$ and $B\bef A$, we need some structure on the formulas in a sequent, i.e. on multisets of formulas, and operations on partial orders. This is a pendant to what happen with cyclic linear logic \cite{Yet90}: formulas are organised in a total cyclic order, and binary rules need operations to combine cyclic orders.  

\subsection{Looking for structured sequents} 
\label{structsequents}  

The formulas of pomset logic we consider are defined from atoms (propositional variables or their negation) by means of the usual commutative multiplicative connectives $\pa$ and $\ts$ together with the new non commutative connective $\bef$ (before)--- the three of them are associative.  
As said above, as De Morgan laws allow, it is assumed that formulas are always in negative normal form.

We want to deal with series parallel partial orders of formulas: $O_1 \pa O_2$ corresponds to parallel composition of partial orders (disjoint union) 
and $O_1 \bef O_2$ corresponds to the series composition of partial orders (every formula in the first partial order $O_1$ is lesser than every formula in the second partial order $O_2$). 
Thus, a formula written with $\pa$ and $\bef$ corresponds to a partial order between its atoms. Unsurprisingly,  we firstly need to study a bit partial orders defined with series and parallel composition. 

However, what about the multiplicative, conjunction namely the $\ts$ connective? It is commutative, but it is distinct from $\pa$. In order to include $\ts$ in this view, where formulas are binary relations on their atoms,  we  consider,  the more general class of irreflexive binary relations that are obtained by $\pa$ parallel composition, $\bef$ series composition and $\ts$ symmetric series compositions, which basically consists in adding the relations of $R_1\bef R_2$
and the ones of $R_2\bef R_1$. The relations that are defined using $\pa$, $\ts$, $\bef$ are called directed cographs or dicographs for short. 

If only $\pa$ and $\ts$ are used the relations obtained are cographs. They have already been quite useful for studying MLL,  see e.g. Theorem \ref{rewMLL} thereafter. 

Before defining pomset logic, we need a presentation of directed cographs.

\subsection{Directed cographs or dicographs} 
\label{dicographs} 

An irreflexive relation  $R\subset P^2$
may be viewed as a graph with vertices $P$ and with both directed edges and undirected edges but without loops. Given an irreflexive relation $R$ let us call its directed part (its arcs)  $\arc R=\{(a,b)\in R| (b,a)\not\in R\}$ and its symmetric part (its edges) $\bar R=\{(a,b)\in R| (b,a)\in R\}$. 
It is convenient to note $a\edge b$ for the edge or pair of arcs $(a,b),(b,a)$ in $\bar R$ 
and to denote $a\fle b$ for $(a,b)$ in $R$ when $(b,a)$ is not in $R$.

\begin{definition}[dicograph]
We consider the class of directed cographs, dicographs for short, which is the smallest class of binary irreflexive relations containing the empty relation on the singleton sets and closed under the following operations defined on pairs of cographs with disjoint domains $E_1$ and $E_2$ yielding a binary relation on $E_1 \uplus E_2$
\begin{itemize} 
\item symmetric series composition $R_1\tssp R_2=R_1\uplus R_2 \uplus (E_1\times E_2) \uplus (E_2\times E_1)$
\item directed series composition $R_1\befsp R_2=R_1\uplus R_2 \uplus (E_1\times E_2)$
\item parallel composition $R_1\pasp R_2=R_1\uplus R_2$
\end{itemize} 
When directed series composition is not used, the graph is said to be a \emph{cograph}. 
When symmetric series composition is not used, the graph is said to be a  \emph{series-parallel partial order} (\sp-order)
\end{definition}

Whenever there are no directed edges (a.k.a.  arcs) the dicograph is a cograph ($\befsp$ is not used). Cographs are characterised by the absence of $P_4$ as many people (re)discovered including us \cite{Ret96entcs}, see e.g. \cite{Kel85}. The graph $P_4: a\edge b \edge c \edge d$ is a path of length $4$, and $P4$-free means that the restriction of the graph to four distinct points never is $P4$: either it contains another edge, or does not contains the three consecutive edges. 

Whenever there are only directed edges (a.k.a.  arcs) the dicograph is an \sp\ order ($\tssp$ is not used) and there are characterised as $N$-free partial orders --- as rediscovered in \cite{Ret93}, see e.g. \cite{Moh85}. The finite order $N$ is  $a <b, c<b, c<d$, and $N$-free means that the restriction of the order to four distinct points never is $N$: either it contains another order relation, or it does not contains the three order relations of $N$. 

We characterised the class of directed dicographs as follows \cite{BGR97,Ret98roma,Ret99romarr}: 

\begin{theorem}
An irreflexive binary relation $R$ is a dicographs if and only if: 
\begin{itemize} 
\item $\arc R$ is N-free ($\arc R$ is an \sp\ order).  
\item $\bar R$  is $P_4$-free ($\bar R$ is a cograph). 
\item Weak transitivity:\footnote{The definition we give here is ours \cite{BGR97}  from 1997. In 1999, Guglielmi found an alternative but equivalent definition \cite{Gug99}.}  \newline 
forall $a,b,c$ in the domain of $R$\newline 
if $(a,b)\in \arc R$ and $(b,c)\in R$ then $(a,c)\in R$ and \newline 
if $(a,b)\in R$ and $(b,c)\in \arc R$ then $(a,c)\in R$
\end{itemize} 
A dicograph can be described with a term in which each element of the domain appears exactly once. This term is written with the three binary operators $\tssp$,  $\pasp$ and $\befsp$ and for a given dicograph this term is unique up to the associativity of the three operators, and to the commutativity of the first two, namely $\pasp$ and $\tssp$. 
\end{theorem}

\begin{definition} 
The \textbf{dual} (or negation) $R^\perp$ of a dicograph $R$ on $P$ is defined as follows: points are given a $\perp$ superscript,
$\arc R^\perp=\arc R$ and  $(\overline{R^\perp})=(P^2 \setminus \bar R) \setminus \{(x,x)| x\in P\}$ or $(a^\perp)=(a)^\perp$, $(a^\perp)^\perp=a$, $(X\tssp Y)^\perp =(X^\perp\pasp Y^\perp)$, $(X\pasp Y)^\perp =(X^\perp\tssp Y^\perp)$, $(X\befsp Y)^\perp =(X^\perp\pasp Y^\perp)$. 
\end{definition} 

\begin{definition}[equivalent poitns in a dicograph]  
Two points $a$ and $b$ of $P$ are said to be \textbf{equivalent} w.r.t. a relation whenever for all $x\in P$ with $x\neq a, b$ 
one as $(x,a)\in R \Leftrightarrow (x,b)\in R$ and $(a,x)\in R \Leftrightarrow (b,x)\in R$. There are three kinds of equivalent points: 
\begin{itemize} 
\item 
Two points $a$ and $b$ in a dicograph are said to be \textbf{freely equivalent} in a dicograph (notation $a \simpa b$) whenever the term can be written (using associativity of $\pasp$ and $\befsp$ and the commutativity of $\pasp$ and $\tssp$)  $T[a\pasp b]$. 
In other words, $a \sim b$,  $(a,b)\not\in R$, $(b,a)\not\in R$. 
\item 
Two points $a$ and $b$ in a dicograph are said to be \textbf{arc equivalent}  in a dicograph (notation $a \simbef b$) whenever the term can be written (using associativity of $\pasp$, $\tssp$ and $\befsp$ and the commutativity of $\pasp$ and $\tssp$)  $T[a\befsp b]$. 
In other words, $a \sim b$,  $(a,b)\in R$, $(b,a)\not\in R$. 
\item 
Two points $a$ and $b$ in a dicograph are said to be \textbf{edge equivalent} in a dicograph (notation $a \simts b$) whenever the term can be written (using associativity of $\pasp$ and $\befsp$ and the commutativity of $\pasp$ and $\tssp$)  $T[a\tssp b]$. 
In other words, $a \sim b$,  $(a,b)\in R$, $(b,a)\in R$. 
\end{itemize} 
\end{definition} 

\subsection{Dicograph inclusion and (un)folding} 
\label{dicographfoldunfold} 

The order on a multiset of formulas, can be viewed as a set of contraints. Hence, when a sequent is derivable with an \sp\ order $I$ it is also derivable with a sub \sp\ order $J\subset I$ --- we named this structural rule entropy \cite{Ret93}. All but one ($\ts\pa4$) of the following transformations of a dicograph into a subdicograph  (w.r.t. inclusion)  preserve provability. Hence we need to characterise the inclusion of a dicograph into another and possibly to view the inclusion as a computational process that can be performed step by step. Fortunately, in \cite{BGR97} we characterised the inclusion of a dicograph in another dicograph by a rewriting relation: 

\begin{theorem} 
A dicograph $R'$ is included into a dicograph $R$ if and only if the term $R$ 
rewrites to the term $R'$ using the rules of Figure \ref{subdicograph} — up to the associativity of $\tssp$, $\befsp$ and $\pasp$, and to the commutativity of $\pasp$ and $\tssp$. 
\end{theorem} 

\begin{figure} 
\label{subdicograph} 
$$
\begin{array}{r@{\ \ }ccccccc @{\ \ }c @{\ \ }cccccccc} 
rule\ name & \multicolumn{7}{c}{dicograph} & \rightsquigarrow & \multicolumn{7}{c}{dicograph'} \\[1.5ex]  
\hline\\[-1ex]  
\skull\quad \tssp\pasp4 & (X &\pasp& Y)  &\tssp& (U &\pasp& V) & \rightsquigarrow &(X &\tssp& U) &\pasp& (Y &\tssp& V)\\[1.5ex]  
\ts\pa3&  (X &\pasp& Y)  &\tssp& U & &  & \rightsquigarrow &(X &\tssp& U) &\pasp& Y && \\[1.5ex]  
\ts\pa2  &   && Y  &\tssp& U &&  & \rightsquigarrow &&& U&\pasp& Y && \\[1.5ex]  
\hline
\\[-1ex]   
\ts{\bef}4 & (X&\befsp& Y)  &\tssp& (U &\befsp& V) & \rightsquigarrow &(X &\tssp& U) &\befsp& (Y&\tssp& V)\\[1.5ex]  
\ts{\bef}3l& (X&\befsp& Y)  &\tssp& U & & & \rightsquigarrow &(X &\tssp& U) &\befsp& Y & & \\[1.5ex]   
\ts{\bef}3r &&& Y &\tssp& (U &\befsp& V) & \rightsquigarrow &  && U &\befsp& (Y &\tssp& V)\\[1.5ex]  
\ts{\bef}2  & && Y  &\tssp& U &&  & \rightsquigarrow & && U &\befsp& Y && \\[1.5ex]  
\hline\\[-1ex]   
\bef\pa4 & (X&\pasp& Y)  &\befsp& (U &\pasp& V) & \rightsquigarrow &(X &\befsp& U) &\pasp& (Y&\befsp& V)\\[1.5ex]  
\bef\pa3l& (X&\pasp& Y)  &\befsp& U & & & \rightsquigarrow &(X &\befsp& U) &\pasp& Y & & \\[1.5ex]  
\bef\pa3r &&& Y &\befsp& (U &\pasp& V) & \rightsquigarrow &  && U &\pasp& (Y &\befsp& V)\\[1.5ex]  
\bef\pa2  & && Y  &\befsp& U &&  & \rightsquigarrow & && U &\pasp& Y && \\[1.5ex]  
\end{array} 
$$ 
\caption{A complete rewriting system for dicograph inclusion. Beware that the first rule $\tssp\pasp4$ marked with a $\skull$ symbol is wrong when the rewriting rule is viewed as a linear implication on formulas: $ (X \pasp Y)  \tssp (U \pasp V)  \not\multimap (X \tssp U) \pasp (Y \tssp V)$ although all other rewriting rules are correct when viewed as linear implications.   } 
\end{figure}

\subsection{Folding and unfolding pomset logic sequents} 
\label{foldunfoldsequent} 

A structured sequent of pomset logic (resp. of MLL) is a multiset of formulas of pomset logic (resp. of MLL)  with the connectives $\bef,\ts,\pa$ endowed with a dicograph. 

\begin{definition}[Folding/Unfolding]  
On such sequents one may define ``folding" and ``unfolding" 
which transform a dicograph of formulas into another dicograph of formulas by combining two equivalent formulas $A$ and $B$ of the dicograph into one formula  $A*B$ (folding) or by splitting one compound formula $A*B$ into its two  immediate subformulas $A$ and $B$ with $A$ and $B$ equivalent in the dicograph.  More formally: 
\begin{description} 
\item[\bf Folding]  Given a multiset of formulas $X_1,\ldots, X_n$ endowed with a dicograph $T$,
\begin{description} 
\item[$\pa$] if $X_i\simpa X_j$ in $T$ rewrite $T[X_i\pasp X_j]$ into $T[(X_i\pa X_j)]$ --- in the multiset, the two formulas $X_i$ and $X_j$ have been replaced with a single  $X_i\pa X_j$. 
\item[$\bef$] if $X_i\simbef X_j$ in $T$ rewrite $T[X_i\befsp X_j]$ into $T[(X_i\bef X_j)]$ --- in the multiset, the two formulas $X_i$ and $X_j$ have been replaced with a single formula $X_i\bef X_j$. 
\item[$\ts$] if $X_i\simts X_j$ in $T$ rewrite $T[X_i\tssp X_j]$ into $T[(X_i\ts X_j)]$ --- in the multiset, the two formulas $X_i$ and $X_j$ have been replaced with a single formula $X_i\ts X_j$. 
\end{description}
\item[\bf Unfolding] is the opposite: 
\begin{description} 
\item[$\pa$] turn $T[(X_i\pa X_j)]$  into $T[X_i\pasp X_j]$  --- in the multiset, the formula  $X_i\pa X_j$ has been replaced with two formulas $X_i$ and $X_j$ with $X_i\simpa X_j$
\item[$\bef$] turn $T[(X_i\bef X_j)]$  into $T[X_i\befsp X_j]$  --- in the multiset, the formula $X_i\pa X_j$ has been replaced with two formulas $X_i$ and $X_j$ with $X_i\simbef X_j$
\item[$\ts$] turn $T[(X_i\ts X_j)]$  into $T[X_i\tssp X_j]$  --- in the multiset, the formula $X_i\pa X_j$ has been replaced with two formulas $X_i$ and $X_j$ with $X_i\simbef X_j$ 
\end{description}
\end{description} 
\end{definition} 

\subsection{A sequent calculus attempt with \sp\ pomset of formulas} 
\label{seqattempts}

Now let us try to extend multiplicative linear logic with a non commutative multiplicative self dual connective (rather than to restrict existing connective to be non commutative), and let us also try to deal with partially ordered multisets of formulas, with $A\bef B$ corresponding to "the subformula $A$ (a resource) comes before the subformula $B$ (another resource)". 

That way one may think of an order on computations: 
\begin{verse} 
a cut between $(A\bef B)^\perp$ and $A^\perp \bef  B^\perp$ reduces to two smaller cuts $A{\--}cut{\--}A^\perp$ and $B{\--}cut{\--}B^\perp$ with the cut on $A$ being  prior to the cut on $B$, while 

a cut between $(A\pa B)^\perp$ and $A^\perp \ts B^\perp$ reduces to two smaller cuts $A{\--}cut{\--}A^\perp$ and $B{\--}cut{\--}B^\perp$ with the cut on $A$ being  in parallel  with the cut on $B$. 
\end{verse} 
This makes sense when linear logic proofs are viewed as programs and cut-elimination as computation. 

Doing so one may obtain a sequent calculus using partially ordered multisets of formulas as in \cite{Ret93} 
but if one wants a sequent with several conclusions that are partially ordered to be equivalent to a sequent with a unique conclusion, one has to only consider \sp\ partial orders of formulas, as defined in Subsection \ref{dicographs} with  parallel composition noted $\pasp$ and series composition noted $\befsp$. 

If we want all formulas in the sequent to be ordered 
the  calculus should handle right handed sequents i.e. be classical.\footnote{Lambek calculus is intuitionistic  and when it is turned into a  classical systems, formulas are endowed with a cyclic order,\cite{Yet90,Abru91,LR96}, i.e. a ternary relation which is not an order and which is quite complicated when  partial --- see the "seaweeds" that first appeared in \cite{Ruet97} and subsequently  used by Abrusci and  Ruet, \cite{AR99} and by de Groote and Lamarche \cite{dGLam2002}.} 

As seen above, we can represent this \sp-ordered multiset of formulas endowed with an \sp\ order by an \sp\ term whose points  are the formulas and such a term is unique up to the commutativity of $\pasp$ and the associativity of $\pasp$ and $\befsp$. 

% modification 27 sept CECI:  
This simple calculus is just here to suggest how one may inductively construct proofs in a well-established framework. 
One should not be too demanding, we mainly ask for this calculus to only handle \sp\ orders on conclusions\footnote{An alternative rule for $\ts$ with \sp\ orders is to apply $\ts$ rule between two minimum formulas in their order component, and to have cut between two formulas one of which is isolated in its ordered sequent. These alternative rules are  trickier and up to our recent investigation this trickier sequent calculus does not enjoy better properties than the  simple seuqent calculus  given in Figure \ref{sequentsporders}.} so  they can be represented with a formula, and to yield correct proof nets only. This sequent calculus is much more restricted than the rules of sequent calculus in \cite{Ret97tlca} or \cite{Ret93}, dealing with general partial orders, and whose aim was to provide a complete calculus yielding all the proof nets, while today Slavnov proposed such a complete sequent calculus \cite{Slavnov2019}.

In this sequent calculus, cuts  are conclusions of the form $K\ts K^\perp$ and they can be part of the order on conclusions: a cut $\Gamma[((A\ts B)\ts (A^\perp\pa B^\perp)]$ reduces into two cuts that are equivalent and parallel in the order  $\Gamma[((A\ts A^\perp)\pasp (B\ts B^\perp)]$, while a cut $\Gamma[((A\bef B)\ts (A^\perp\bef B^\perp)]$
 reduces into two cuts that are equivalent but come one before the other $\Gamma[((A\ts A^\perp)\befsp (B\ts B^\perp)]$.
 This will be better explained in the proof net section.

\begin{figure} 
\begin{center} 
\input{pomset_orders.tex}
\end{center} 
\caption{Sequent calculus on \sp\ pomset or formulas; called \sp-pomset sequent  calculus}
\label{sequentsporders} 
\end{figure}

\section{Proof nets} 
\label{pn} 

This section presents proof structures and nets (the correct proof structures), in an abstract and algebraic 
manner, without links nor trip conditions: such proof structures and nets are called handsome proof structures and nets.  Basically proof nets consists in a dicograph $R$ of atoms representing the conclusion formula, and axioms that are disjoint pairs of dual atoms constituting a partition $B$ of the atoms of $R$. 
The proof net can be viewed as an edge bi-coloured graph:  the dicograph is represented by $R$ arcs and edges (Red and Regular in the pictures),  while the axioms $B$ (Blue and Bold in the picture). 
In such a setting, the correctness criterion expresses some kind of orthogonality between $R$ and $B$. A proof net can also be viewed as a term, axioms being denoted by indices used exactly twice on dual atoms.

\subsection{Handsome pomset proof nets} 
\label{handsomepn} 

In fact, proof nets have (almost) been defined above! 

\begin{definition} 
A pomset logic \emph{handsome proof structure} or \dicogpn\ is a graph $G= (V, B, R)$ with two kind of edges: 
\begin{itemize} 
\item V: $\{a_1,\ldots,a_n,a_1^\perp,\ldots,a_n^\perp\}$ i.e. vertices consist in a finite number $n$ of  pairs of dual atoms,  --- some of which may have the same name, hence they should be distinguish with a superscript.
\item B edges (B stands for Bold or Blue) are undirected and link a vertex to its dual. B edges define a perfect matching of G that is to say no two B edges are adjacent and every vertex is incident to a B edge .
\item R edges (R for Regular or Red) are dicograph over the vertices. There are no loops, somme R edges are directed and are called arcs, some are undirected. The dicograph can be described by a term, but when terms that only differ because of the associativity of $\pa,\bef,\ts$ and the commutativity of $\pa,\ts$, the proof structures $\pi$ and $\pi'$ are \emph{equal}. 
\end{itemize} 
\end{definition} 

Of course, not all proof structures are correct, for instance $(\{a,a^\perp\}, B=\{a,a^\perp\}, R=a \tssp a^\perp)$ is incorrect. 

\begin{criterion} 
\label{criterionHandsome} 
A handsome proof structure is said to be a handsome proof net or to be correct whenever every elementary circuit (directed cycle) of alternating edges in $R$ and in $B$ contains a chord --- an edge or arc connecting two points of the circuit but not itself nor its reverse in the circuit. 
In short, every \ael\ circuit contains a chord. Observe that this chord cannot be in $B$, hence it is in $R$, and it can either be an $R$ arc or an $R$ edge. 
\end{criterion}

\begin{theorem}[Nguy\^en] 
Recently it was established that checking whether a proof structure satisfies the above correctness  criterion is coNP complete \cite{nguyen2019proof}. 
\end{theorem} 

\begin{figure}
    \centering
\subfigure[Two incorrect handsome proof structures (chordless \ae-circuit: {$a,c^\perp,c,a^\perp,a$} in both cases)]{
\hfill
\hspace*{1cm}
\tikzset{vertex/.append style={inner sep=2pt}} %% Sylvain: enlarge the circles for vertices
 \begin{tikzpicture}[PS]
  \vertex[orientation=180,on circle]{a};
    \vertex[orientation=120,neg,on circle]{a};
   \vertex[orientation=60,on circle]{b};
   \vertex[orientation=0,neg,on circle]{b};
   \vertex[orientation=-60,on circle]{c};
   \vertex[orientation=-120,neg,on circle]{c};
   \blink{c}; 
   \blink{b};
   \blink{a};
    \draw[Rlink,->] (a) -- (b);
    \draw[Rlink,->] (acomp) -- (b);
    \draw[Rlink,->] (c) -- (b);
    \draw[Rlink,->] (ccomp) -- (b);
   \draw[Rlink,->] (a) -- (bcomp);
    \draw[Rlink,->] (acomp) -- (bcomp);
    \draw[Rlink,->] (c) -- (bcomp);
    \draw[Rlink,->] (ccomp) -- (bcomp);
   %  \draw[Rlink] (a) -- (c);
    \draw[Rlink] (a) -- (ccomp);
     \draw[Rlink] (acomp) -- (c);
%    \draw[Rlink] (acomp) -- (ccomp);   
 \end{tikzpicture} 
\hspace{1cm}
 \begin{tikzpicture}[PS]
  \vertex[orientation=180,on circle]{a};
    \vertex[orientation=120,neg,on circle]{a};
   \vertex[orientation=60,on circle]{b};
   \vertex[orientation=0,neg,on circle]{b};
   \vertex[orientation=-60,on circle]{c};
   \vertex[orientation=-120,neg,on circle]{c};
   \blink{c}; 
   \blink{b};
   \blink{a};
    \draw[Rlink,->] (a) -- (b);
    \draw[Rlink,->] (acomp) -- (b);
    \draw[Rlink,->] (c) -- (b);
    \draw[Rlink,->] (ccomp) -- (b);
   \draw[Rlink,->] (a) -- (bcomp);
    \draw[Rlink,->] (acomp) -- (bcomp);
    \draw[Rlink,->] (c) -- (bcomp);
    \draw[Rlink,->] (ccomp) -- (bcomp);
   %  \draw[Rlink] (a) -- (c);
    \draw[Rlink,->] (a) -- (ccomp);
     \draw[Rlink,<-] (acomp) -- (c);
%    \draw[Rlink] (acomp) -- (ccomp);   
 \end{tikzpicture}
\hspace*{1cm}
\hfill}\\
\subfigure[Two correct handsome  proof structures
(i.e. two correct hansome proof nets)]{
\hspace*{1cm}
\hfill
\begin{tikzpicture}[PS]
  \vertex[orientation=180,on circle]{a};
    \vertex[orientation=120,neg,on circle]{a};
   \vertex[orientation=60,on circle]{b};
   \vertex[orientation=0,neg,on circle]{b};
   \vertex[orientation=-60,on circle]{c};
   \vertex[orientation=-120,neg,on circle]{c};
   \blink{c}; 
   \blink{b};
   \blink{a};
    \draw[Rlink,->] (a) -- (b);
    \draw[Rlink,->] (acomp) -- (b);
    \draw[Rlink,->] (c) -- (b);
    \draw[Rlink,->] (ccomp) -- (b);
   \draw[Rlink,->] (a) -- (bcomp);
    \draw[Rlink,->] (acomp) -- (bcomp);
    \draw[Rlink,->] (c) -- (bcomp);
    \draw[Rlink,->] (ccomp) -- (bcomp);
     \draw[Rlink] (a) -- (c);
    \draw[Rlink] (a) -- (ccomp);
     \draw[Rlink] (acomp) -- (c);
    \draw[Rlink] (acomp) -- (ccomp);   
 \end{tikzpicture} 
\hspace{1cm}
 \begin{tikzpicture}[PS]
  \vertex[orientation=180,on circle]{a};
    \vertex[orientation=120,neg,on circle]{a};
   \vertex[orientation=60,on circle]{b};
   \vertex[orientation=0,neg,on circle]{b};
   \vertex[orientation=-60,on circle]{c};
   \vertex[orientation=-120,neg,on circle]{c};
   \blink{c}; 
   \blink{b};
   \blink{a};
    \draw[Rlink,->] (a) -- (b);
    \draw[Rlink,->] (acomp) -- (b);
    \draw[Rlink,->] (c) -- (b);
    \draw[Rlink,->] (ccomp) -- (b);
   \draw[Rlink,->] (a) -- (bcomp);
    \draw[Rlink,->] (acomp) -- (bcomp);
    \draw[Rlink,->] (c) -- (bcomp);
    \draw[Rlink,->] (ccomp) -- (bcomp);
   %  \draw[Rlink] (a) -- (c);
    \draw[Rlink,->] (a) -- (ccomp);
     \draw[Rlink,->] (acomp) -- (c);
%    \draw[Rlink] (acomp) -- (ccomp);   
 \end{tikzpicture}
 \hspace*{1cm}
 \hfill
} 
\caption{Two incorrect handsome proof structures on top, two correct handsome proof structures (or proof nets) underneath.} 
\end{figure}

\begin{theorem}\label{rewPreservePN} 
Given a proof net $(V, B, R)$ if $R\rightsquigarrow R'$ (so $R'\subset R$) using rewriting  rules of Figure \ref{subdicograph} \underline{except $\ts\pa4$}  then $(B,R')$ is a proof net as well, i.e. all the rewrite rule preserve the correctness criterion on page \pageref{criterionHandsome}.
\end{theorem} 

\begin{proof} 
See \cite{Ret98roma,Ret99romarr}. 
\end{proof} 

\begin{corollary}
\label{cutasrew} 
If a proof structure $(V,B,R)$ with 
$$R=t[K(x_1,...,x_k)\tssp K^\perp(x_1^\perp,...,x_k^\perp)]$$ is correct, so is the proof structure $(V,B,R')$ with $$R'=t[(x_1\tssp x_1^\perp)\pasp \cdots \pasp (x_k\tssp x_k^\perp)]$$
In other words, rewriting a cut of a correct handsome  proof net into smaller cuts preserves the correctness criterion. 
\end{corollary} 

\begin{proof} We proceed by induction on the size of the cut-formula $K$ to show that $R'$ is actually obtain by rewriting not using $\ts\pa4$, and given Theorem \ref{rewPreservePN} the result holds. 

If $K$ contains no connective, then $k=1$ and there is nothing to prove. 

\medskip 

If $K=K_1\tssp K_2$ then $K^\perp=K_1^\perp\pasp K_2^\perp$ and\\   
\centerline{$(K_1\tssp K_2) \tssp (K_1^\perp\pasp K_2^\perp)= (K_1\tssp  (K_1^\perp\pasp K_2^\perp))\tssp K_2$}\\ 
which rewrites to $(  (K_1\tssp K_1^\perp)\pasp K_2^\perp)\tssp K_2$ by $\ts\pa3$
\\ 
which rewrites to $(  (K_1\tssp K_1^\perp)\pasp (K_2^\perp\tssp K_2))$ by $\ts\pa3$. 
\medskip

If $K=K_1\bef K_2$ then $K^\perp=K_1^\perp\befsp K_2^\perp$ and  $(K_1\bef K_2) \tssp (K_1^\perp\pasp K_2^\perp)$
\\ 
which rewrites to $(  (K_1\tssp K_1^\perp)\befsp (K_2 K_2^\perp)$ by $\ts\bef5$
\\
which rewrites to $(  (K_1\tssp K_1^\perp)\pasp (K_2^\perp\tssp K_2))$ by $\bef\pa2$. 

\medskip

In both cases we end up with a similar situation with $K_1, K_2$ having one connective less. 
\end{proof}

\begin{proposition} 
\label{tspa4incorrect}
In some cases $\ts\pa4$ does not preserve correctness, i.e. turns a correct proof net  into an incorrect proof structure. 
\end{proposition}

\begin{proof} 
Consider the following example: 
$B=\{a\edge a^\perp, b\edge b^\perp\}$
$R=\bar B=(a\pa a^\perp)\ts(b\pa b^\perp)=  \{a\edge b,a\edge b^\perp,
a^\perp\edge b^\perp, 
a^\perp\edge b\}$.  

Using $\ts\pa4$, $R=\bar B$ rewrites to\newline  $R'=(a\ts b)\pa (a^\perp\ts b^\perp)=\{a\edge b,a^\perp\edge b^\perp\}$, 
and the proof net 
$(B,R')$ contains the chordless \ae\ circuit $(a,a^\perp)\in B,(a^\perp,b^\perp)\in R', (b^\perp,b)\in B, (b,a)\in R$. 
\end{proof} 

Observe that it does \emph{not} mean that every correct proof net $(V,B,R)$ with axioms $B$ can be obtained from $(V,B,R=\bar B)$ by the allowed rewrite rules (all but $\ts\pa4$) where $\bar B$ is 
$\tssp_i (a_i\pasp a_i^\perp)$. Indeed, since $R\subset \bar B$ it is known that $\bar B\rightsquigarrow R$ but one cannot tell whether $\ts\pa4$ has been used.

Indeed, as shown above $\ts\pa4$ does not preserve correctness  but it may happen, and all the rules (including $\ts\pa4$)
do not preserve incorrectness, i.e. may turn an incorrect proof structure into a correct proof net.

\subsection{Cut and cut-elimination} 
\label{handsomecuts} 

What about the cut rule? This calculus has no rules in the standard sense, in particular no binary rules that would combine a $K$ and a $K^\perp$. 
A cut is a tensor $K\ts K^\perp$ which may not be a strict subformula of some other formula. 

So a cut in this setting simply is a symmetric series composition $K \tssp K^\perp$ in a dicograph whose form is $T \pasp (K \tssp K^\perp)$.  Assume the atoms of $K$ are $\{a_1,\ldots,a_n\}$ so atoms of $K^\perp$ are $\{a_1^\perp,\ldots,a_n^\perp\}$. 
Cut-elimination consist in suppressing  all edges and arcs 
between two atoms of $K$, all edges and arcs between two atoms of $K^\perp$, and all edges 
$a_i,a_j^\perp$ with $i\neq j$ --- so the only edges incident to $a_i$ are $a_i,a_i^\perp$ (call those edges atomic cuts) and $a_i x$ with $x$ neither in $K$ not in $K^\perp$. This yields a correct proof net because of Corollary \ref{cutasrew}.

If, in  this graph, an atom $a$ is in the $B$ relation with an $a^\perp$ in $K\cup K^\perp$, then  the result of cut elimination is the closest point not in $K$ nor in $K'$ reached by an alternating sequence of $B$-edges  and elementary cuts starting from $a$ -- observed that this point is necessarily named $a^\perp$, that we call its \emph{cut neighbour}.

To obtain the proof resulting from cut-elimination suppress all the atoms of $K$ and $K^\perp$
as well as the incident arcs and edges and connect  every atom to its cut neighbour with a $B$ edge. 

The B edges from an atom to its cut neighbour can be obtained step by step by replacing a path of a B edge and R edge and a B edge with a B edge, and this preserves correctness (it is related to the rule $a{\uparrow}$ of SBV discussed in Section \ref{PomsetMLLrewriting}. 

This final step leads us to the cut elimination theorem for handsome pomset proof nets: 

\begin{theorem} 
Cut elimination preserves the correctness criterion of \dicogpn\ proof nets and consequently the f \dicogpn\ proof nets enjoy cut-elimination. 
\end{theorem} 

\begin{proof} The preservation of the absence of chordless \ae\ circuit during cut elimination is proved in \cite{Ret98roma,Ret99romarr}. There are two steps in cut-elimination: one consists in turning the cut into atomic cuts, which is proved to preserve correctness in Corollary \ref{cutasrew} and the other consists in shortening succession of atomic cuts and axioms according to the following pattern: 
$a^\perp_i B a_j R cut a^\perp_k B a_l $ reduces to $a^\perp_i B a_l$ provided the $R$ edges belong to the cut (we here distinguish the vertices, but their names as  logical atoms   they are all equal or dual,  $a_i=a_j=a_k=a_l$, because they are linked by cuts or axioms). Such a suppression of \ae paths without chords cannot create chordless \ae circuits. 
\end{proof}

\subsection{From sequent calculus and rewrite proofs to \dicogpn}\label{pomsetrewrite} 

Proofs of the sequent calculus given in Figure \ref{sequentsporders} 
are easily turned into a \dicogpn\ proof net inductively. 
Such a derivation starts with axioms $\seq a_i, a_i^\perp$ as it is well known, and in any kind of multiplicative linear logic the atoms $a_i$ and  $a_i^\perp$
that can be traced from the axiom that introduced them to the conclusion sequent, which, after some  unfolding can be viewed as a dicograph of atoms $R$.
The \dicogpn\ proof structure corresponding to the sequent calculus proof simply is  $(B,R)$, and fortunately is a correct proof net. 

\begin{proposition} 
A proof of sequent calculus corresponds to a \dicogpn\ i.e. to a handsome proof structure without chordless alternate elementary path, i.e. into a handsome proof net. 
\end{proposition} 

\begin{proof}  In \cite{Ret98roma,Ret99romarr} we established by induction on the proof, we that  neither the folding rules nor the unfolding rules can introduce a chordless \ael\ cycle. 
\end{proof} 

% modification 27 sept 
%The above result also yields cut elimination for the sequent calculus. Indeed, when a proof net $\pi$ is issued from a sequent calculus proof $s(\pi)$, any reduced proof net $\pi'$  obtained by cut elimination from $\pi$, does also correspond to a sequent calculus proof. 
% fin modification 27 sept 

\begin{proposition} 
Any proof obtained by rewriting from $AX_n$ yields a 
handsome proof structure 
without chordless alternate elementary path, i.e. into a \dicogpn. 
\end{proposition} 

\begin{proof} Observe that $AX_n$ satisfies the criterion, so because of Theorem \ref{rewPreservePN}, hence the result is clear. \end{proof}

\section{Denotational semantics of pomset logic within coherence spaces}
\label{coh}

Denotational semantics or categorical interpretation of a logic is the interpretation of a logic in such a way that a proof $d$ of $A\seq B$ is interpreted as a morphism $\llbracket d \rrbracket $ from an object $\llbracket A \rrbracket$ to an object $\llbracket B \rrbracket$ in such a way that $\llbracket d \rrbracket = \llbracket d' \rrbracket$ whenever $d$ reduces to $d'$ by  (the transitive closure) of $\beta$-reduction  or cut-elimination. 
A proof $d$ of $\seq B$ (when there is no $A$) is simply interpreted as a morphism from the terminal object $\mathbf{1}$ to $B$. More details can be found in  \cite{LS86,GLT88}.  

Once the interpretation of propositional variables is defined, 
the interpretation of complex formulas is defined 
by induction on the complexity of the formula. 
The set $Hom(A,B)$ of morphisms  from $A$ to $B$ is in bijective correspondence with an object written $\llbracket A\multimap B\rrbracket$. 
Morphisms are defined by induction on the proofs and one has to check that 
the interpretations of a proof before and after one step of cut elimination is unchanged.

Categorical interpretaitons of  intuitionistic logic,
take places in Cartesian closed categories while categorical intepretations of   take place in a monoidal closed category (with  monads for multiplicative exponential linear logic).

\subsection{Coherence spaces}
\label{cohDef}

The category of coherence spaces is a concrete category: objects are (countable) sets endowed with a binary relation, and morphisms are linear maps.  It interprets the proofs up to cut-elimination or $\beta$ reduction initially  propositional intuitionistic logic and propositional linear logic (possibly quantified). Actually, coherence spaces are tightly related to linear logic: indeed, linear logic arose from this particular semantics, invented to model second order lambda calculus i.e. quantified propositional intuitionistic logic \cite{Gir86}. Coherence spaces are themselves inspired from the categorical work on ordinals by Jean-Yves Girard; they are the binary qualitative domains. 

\begin{definition} 
A coherence space $A$ is a set $|A|$ (possibly infinite) called the \emph{web} of $A$  whose elements are called \emph{tokens}, endowed with a binary reflexive and symmetric relation called \emph{coherence} on $|A|\times |A|$ noted $\alpha\coh \alpha'[A]$ or  simply $\alpha\coh \alpha'$ when $A$ is clear. 
\end{definition} 

The following notations are common and useful: 
\begin{verse} 
$\alpha\scoh \alpha'[A]$  iff  $\alpha\coh \alpha'[A]$ and $\alpha\neq \alpha'$

$\alpha\incoh \alpha'[A]$  iff  $\alpha\not\coh \alpha'[A]$ or $\alpha=\alpha'$

$\alpha\sincoh \alpha'[A]$  iff  $\alpha\not\coh \alpha'[A]$ and $\alpha\neq \alpha'$
\end{verse}

A proof of $A$ is to be interpreted by a \emph{clique}  of  the corresponding coherence spaces $A$, a \emph{cliques} being a set of pairwise coherent tokens in $|A|$ --- we write $x\in A$ for $x\subset |A|$ and for all $\alpha,\alpha'$ in $x$ one has $\alpha\coh \alpha'$. Observe that forall $x\in A$, if $x'\subset x$ then  $x'\in A$. 

\begin{definition} 
A linear  morphism $F$ from $A$ to $B$ is a morphism mapping cliques of $A$ to cliques of $B$ such that: 
\begin{itemize} 
\item For all  $x\in A$ if  $(x'\subset x)$ then  $F(x')\subset F(x)$
\item For every family $(x_i)_{i\in I}$ of pairwise compatible cliques --- that is to say    $(x_i\cup x_j)\in A$ holds for all $i,j\in I$ --- 
 $F(\cup_{i\in I} x_i)= \cup_{i\in I} F(x_i)$.\footnote{The morphism is said to be stable when this second condition is replaced with $F(\cup_{i\in I} x_i)= \cup_{i\in I} F(x_i)$ holds more generally for the union of a directed family of cliques of $A$, i.e. $\forall i,j\exists k\ (x_i\cup x_j)=x_k$.}  
\item For all $x,x'\in A$ if $(x\cup x')\in A$ then $F(x\cap x')=F(x)\cap F(x')$ -- thsi last condition is called \emph{stability}.  
\end{itemize} 
\end{definition} 

Due to the removal of structural rules,  linear logic has two kinds of conjunction: 

$$\begin{prooftree} 
\seq \Gamma, A\quad \seq \Delta, B 
\justifies \Gamma,\Delta, A \ts B
\using \ts
\end{prooftree}
\qquad 
\begin{prooftree} 
\seq \Gamma, A\quad \seq \Gamma, B 
\justifies \Gamma, A \& B
\using \& 
\end{prooftree}
$$ 

Those two rules are equivalent when  contraction and weakening are allowed. 
The multiplicatives (contexts are split, $\ts$ above) and the additives (contexts are duplicated, $\&$ above). Regarding  denotational semantics, the web of the coherence space associated with a formula $A * B$ with $*$ a \emph{multiplicative} connective  is  the Cartesian product $|A|\times |B|$ of the webs of  $A$ and $B$ --- while it is the  disjoint union of the webs of  $A$ and $B$ when the connective $*$ is additive. 

Negation is a unary connective which is both multiplicative and additive: 

\centerline{$|A^\perp|=|A|$ and $\alpha\coh '\alpha[A^\perp]$ iff $\alpha\incoh \alpha'[A]$} 

One may wonder how many binary multiplicative connectives there are,
i.e. how many different coherence relations one may define on $|A|\times |B|$ from the coherence relations on $A$ and on $B$. 

We can limit ourselves to the ones that are covariant functors in both $A$ and $B$ --- indeed there is a negation, hence a contravariant connective in $A$ is a covariant connective in $A^\perp$. Hence when both components are $\incoh$ so are the two couples, 
and when they are both coherent, so are the two couples.

As it is easily observed, given two tokens $\alpha,\alpha'$ in a coherence space $C$
exactly one of the three following properties hold: 

$$\alpha \sincoh \alpha' \qquad \alpha = \alpha' \qquad \alpha \scoh \alpha'$$  

To define a multiplicative connective, is to define when $(\alpha,\beta)\coh(\alpha',\beta')[A*B]$ holds depending on whether  $\alpha\coh \alpha'[A]$ and $\beta\coh \beta'[B]$ hold. 
Thus defining a binary mutliplicative connective is to fill a nine cell table as the ones below, the first column inndicates the relation between $\alpha$ and $\alpha'$ in $A$, while the first row indicates  the relation between $\beta$ and $\beta'$ in $B$. 

However if $*$ is assumed to be covariant in both its arguments, seven out of the nine cells are filled, so only free values  are the ones in the right upper cell, NE=North-East,  and the left bottom cell SW=South-West.  They cannot be $=$ so it makes four possibilities. 

\renewcommand\ne{{\scriptstyle\mathrm{NE}}}
\newcommand\sw{{\scriptstyle\mathrm{SW}}}

$$ 
\tablecoh{A}{B}{*}{\ne?}{\sw?}
$$

If one wants $*$ to be commutative, there are only two possibilities, namely $NE=SW=\scoh$ ($\pa$)  and $NE=SW=\sincoh$ ($\ts$).

$$\tablecoh{A}{B}{\pa}{\scoh}{\scoh} 
\hfill 
\mbox{\qquad and \qquad} 
\hfill 
\tablecoh{A}{B}{\ts}{\sincoh}{\sincoh} 
$$

However if we do not ask for the connective $*$ to be commutative we have a third connective $A\bef B$ and a fourth connective $A\triangleright B$ which is simply $B\bef A$.  

$$
\tablecoh{A}{B}{\bef}{\sincoh}{\scoh} 
\mbox{\qquad and \qquad} 
\tablecoh{A}{B}{\triangleright}{\scoh}{\sincoh} 
$$ 

The coherence relation on those connectives are defined as follows:

$(\alpha,\beta)\coh(\alpha',\beta')[A\ts B]$ iff $\alpha\coh \alpha'[A]$ and $\beta\coh \beta'[B]$

$(\alpha,\beta)\coh(\alpha',\beta')[A\pa B]$ iff $\alpha\coh \alpha'[A]$ or  $\beta\coh \beta'[B]$

$(\alpha,\beta)\scoh(\alpha',\beta')[A\bef B]$ iff $\alpha\scoh \alpha'[A]$ or  ($\alpha=\alpha'$ and $\beta\coh \beta'[B]$)

The definition of $\bef$ and $\pa$ coherence spaces applies to \sp\ partial orders of formulas. 
Given  an \sp-order $T[A_1,\ldots,A_n]$ on the formulas $A_1,\ldots,A_n$, i.e.  a dicograph term $T$  using only  $\pa$ and $\bef$ the above defintions of $\bef$ and $\pa$ 

$(\alpha_1,\ldots, \alpha_n)\scoh (\alpha'_1,\ldots, \alpha'_n)[T[A_1,\ldots,A_n]$
are strictly coherent whenever: there exist $i$ such that  $\alpha_i\scoh \alpha'_i$
and for every  $j>i$  one has 
$\alpha_j = \alpha'_j$.

The linear morphisms from $A$ to $B$ that is $Hom(A,B)$ can be represented by the coherence space  $A\multimap B=A^\perp\pa B$ \footnote{This internalisation of $Hom(A,B)$ makes the category closed, but not cartesian closed because the associated conjunction, namely $\ts$ is not a product.}

$$ 
\begin{array}{l||c|c|c|}
A \multimap B & \sincoh & = & \scoh \\ \hline \hline 
\sincoh & \scoh & \scoh & \scoh  \\ \hline 
= & \sincoh & = & \scoh \\ \hline 
\scoh & \sincoh  & \sincoh & \scoh \\ \hline 
\end{array} 
$$ 

Let us see how  linear morphisms are in a one-to-one correspondence with cliques of $A\multimap B$.  Given a clique $F\in (A\multimap B)$ the map $F_f$ from cliques of $A$ to cliques of $B$ defined $F_f(x)=\{\beta\in |B|\ |\ \exists \alpha\in x\ (\alpha,\beta)\in f\}$ is a linear morphism. Conversely, given a linear morphism $F$, the set $\{(\alpha,\beta)\in |A|\times |B|\ |\ \beta\in F(\{\alpha\})\}$ is a clique of $A\multimap B$. 

One can observe that the subset $\{((\alpha,(\beta,\gamma)),((\alpha,\beta),\gamma))\ |\ \alpha\in |A|,\beta\in |B|, \gamma\in |C|\}$ of  $|A|\times |B|\times |C|$ defines a linear isomorphism from $A\bef (B\bef C))$ to 
$(A\bef B)\bef C$, that $\{((\alpha,\beta),(\alpha,\beta))\ |\ \alpha\in |A|, \beta\in |B|\}$
defines a linear morphism from $A\ts B$ to $A\bef B$ and the same set of pairs of tokens also defines a linear morphims from $A\bef B$ to $A\pa B$. However, for general coherence spaces $A$ and $B$ 
there is no canonical linear map from $A\bef B$ to $B\bef A$. 

Given two different tokens $(\alpha,\beta)$ and $(\alpha',\beta')$ in $|A|\times|B|$, observe that: 
\begin{enumerate} 
\item 
$(\alpha,\beta)\scoh(\alpha',\beta')[A\bef B]$ means 
($\alpha=\alpha'$ and $\beta\scoh\beta'[B]$)   or $\alpha\scoh\alpha'[A]$
\item 
$(\alpha,\beta)\scoh(\alpha',\beta')[A^\perp\bef B^\perp]$ means 
($\alpha=\alpha'$ and $\beta\sincoh\beta'[B]$)   or $\alpha\sincoh\alpha'[A]$
\end{enumerate} 

Given that those tow tokens are different, either: 
\begin{verse} 
If $\alpha\neq\alpha'$ then 
either $\alpha\scoh\alpha'[A]$,  1 holds and 2 does not hold or $\alpha\scoh\alpha'[A]$,  2 holds and 1 does not hold.

If $\alpha=\alpha'$ then $\beta\neq\beta'$ and 
either $\beta\scoh\beta'[B]$, 1 holds and 2 does not hold 
or $\beta\sincoh\beta'[B]$ 2 holds and 1 does not hold. 
\end{verse} 
Consequently, $(A\bef B)^\perp\equiv A^\perp\bef B^\perp$. 

Linear logic is issued from coherence semantics,  
and consequently coherence semantics is close to linear logic syntax. Coherence spaces may even be turned into a fully abstract model in the multiplicative case (without before), see \cite{Loa94lics}. 

The before connective is issued from coherence semantics, hence it is a good idea to explore the coherence semantics of the logical calculi we designed for pomset logic, to see whether they are sound. 

\subsection{A sound and faithful interpretation of proof nets in coherence spaces}
\label{cohInt} 

An important  criterion 
comforting the design of the deductive systems for pomset logic is that those systems are sound w.r.t. coherence semantics --- in addition to cut-elimination discussed previously. 
We shall here interprets a proof net with conclusion $T$ (a formula or a dicograph of atoms) as a clique of the corresponding coherence space $T$. 

Computing the semantics of a cut-free proof net is rather easy, using Girard's experiments but from axioms to conclusions as done in \cite{Ret94rc,Ret97}. 

However, we define the interpretation of a proof structure (non necessarily a proof net) as a set of tokens of the web of the conclusion formula. Assume the proof structure is $B=\{a_i\edge a_i^\perp|1\leq i \leq n\}$ and that each of the $a_i$ as a corresponding coherence space $a$ also denoted by $a_i$. 
For each $a_i$ choose a token $\alpha_i\in |a_i|$. If the conclusion is a dicograph $T$ replacing each occurrence of $a_i$ and each occurrence of $a_i^\perp$ with $\alpha_i$ yields a term, which when converting $x * y$ (with $*$ being one of the connectives, $\pa,\bef,\ts$) with $(x,y)$, yields a 
token in the web of the coherence space associated with $T$ --- this token in $|T|$ is called the result of the  experiment.

Given a normal (cut-free) proof structure $\pi$ with conclusion $T$ 
the interpretation $\llbracket \pi \rrbracket$ of the normal proof structure $\pi$
is the set of all the results of the experiments on $\pi$. One has the following result that Lambek appreciated, because it replaces graph theoretical considerations with algebraic properties:

\begin{theorem}\label{semcorrect} 
A proof structure $\pi$ with conclusion $T$ is a proof net (contains no chordless \ael-circuit) if 
and only if its interpretation $\llbracket \pi \rrbracket$ is a clique of the coherence space  $T$ (is a semantic object). 
\end{theorem} 

\begin{proof}
The proof is a consequence of: 
\begin{itemize} 
\item both folding and unfolding (see Subsection \ref{PNling} or \cite{Ret99rr,Ret98roma,Ret99romarr}) preserve correctness 
\item  semantic characterisation of proof nets with links correctness is proved in \cite{Ret94rc} for MLL and pomset logic --- the published version left out pomset logic \cite{Ret97}. 
\end{itemize} 
\end{proof} 

The actual result we proved is a bit more: in order to check correctness a given four-token coherence space is enough, and this provides a  way to check correctness, which is of an exponential complexity in accordance with the recent results by Nguy\^en \cite{nguyen2019proof}.

When $\pi$ is not normal, i.e. includes cuts, not all  experiments succeed and provide results:  an experiment is said to succeed when in every cut $K cut K^\perp$ the value $\alpha$ on an atom $a$ of the cut $K$ is the same as the value on the corresponding atom $a^\perp$ in  $K^\perp$. Otherwise the experiment fails and has no result.  The set of the results of all succeeding experiments of a proof \emph{net}  $\pi$ is a clique of the coherence space $T$. It is the interpretation $\llbracket \pi \rrbracket$ of the normal proof net $\pi$. Whenever $\pi$ reduces to $\pi'$ by cut elimination $\llbracket \pi\rrbracket=\llbracket \pi' \rrbracket$. 
That way one is able to predict whether a proof structure will reduce to a proof net\footnote{Proof nets reduces to proof nets, correctness is preserved under cut-elimination, but an incorrect proof structure may well reduce to a proof net.} without actually performing cut elimination: 

\begin{theorem} 
Let  $\pi$ be a proof structure  and let $\pi^*$ be its normal form;  then 
$\pi^*$ is a proof net plus zero or more loops (cut between two atoms that are connected with an axiom) whenever two succeeding experiments of  $\pi$ have coherent or equal results. 
\end{theorem}

\begin{proof} See \cite{Ret94rc,Ret97}. \end{proof}

\section{Sequentialisation with pomset sequents or dicographs sequents} 
\label{sequentialisation} 

In 2001, Lambek  noticed the absence of sequent calculus in my habilitation \cite{RetoreHDR}. 
Although there is one in my PhD that was refined later to only use \sp\ orders, I did not put sequent calculus on the forefront firstly because the proof net calculus enjoys much more  mathematical properties and secondly because the sequent calculi I know do not generate all the proof nets. I tried, and Sylvain Pogodalla and Lutz Stra{\ss}burger  as well, to prove  that every  correct proof net is  the image of a proof in the sequent calculus --- the one given here or some variant. 

The \sp-pomset sequent calculus presented in Figure \ref{sequentsporders} is clearly equivalent to the dicograph sequent calculus with  dicographs of atoms as sequents; in the dicograph sequent calculus, the symmetric series composiitons $\tssp$ may well be used on contexts, as the $\pasp$ and $\befsp$ rule, and all connective introduction rules consists in internalising the $\widehat{*}$ operation inside a formula as a $*$ connective.  This calculus is shown in Figure \ref{sequentdicographs}. Observe that entropy  does not allow inclusion of dicograph in general, but only of an outer \sp-order; indeed, in general, dicograph inclusion does not preserve correctness, as explained in Proposition \ref{tspa4incorrect}. 

\begin{figure} 
\begin{center} 
% rules for pomset logic 
% with dicograpsh of atoms 
% the "largest" system
% simple premises of tensors are isolated in the order 

\begin{prooftree} 
\justifies 
\seq a \pasp  a^\perp 
\using \mathrm{axiom} 
\end{prooftree} 

\bigskip

\begin{prooftree} 
\seq \Gamma \qquad  \seq \Delta
\justifies 
\seq \Gamma \befsp \Delta
\using dimix
\end{prooftree}

\bigskip 

\begin{prooftree} 
\seq O[\Gamma_1,\ldots,\Gamma_p]
\justifies 
\seq O'[\Gamma_1,\ldots,\Gamma_p]
\using \mathrm{entropy}
\left\{\begin{tabular}[c]{l} \small 
with $\Gamma_i$:dicographs,\\ \small $O,O'$ \sp-orders, $O'\subset O$ 
\end{tabular}\right.
\end{prooftree} 
 
\bigskip 

\begin{prooftree} 
\seq A\pasp \Gamma  \qquad  \seq B\pasp \Delta 
\justifies 
\seq \Gamma \pasp (A\tssp B) \pasp \Delta 
\using \ts \mbox{ / cut when }  A=B^\perp
\end{prooftree} 

\bigskip 

\begin{prooftree} 
\seq \Gamma[A\pasp B]
\justifies 
\seq \Gamma[A\pa B] 
\using \pa \mbox{ if } A\simpa B
\end{prooftree} 
\qquad 
\begin{prooftree} 
\seq \Gamma[A\befsp B]
\justifies 
\seq \Gamma[A\bef B] 
\using \bef \mbox{ if }  A\simbef B
\end{prooftree} 
\qquad 
\begin{prooftree} 
\seq \Gamma[A\tssp B]
\justifies 
\seq \Gamma[A\ts B] 
\using \ts \mbox{ if } A\simts B
\end{prooftree} 
\bigskip

% I\lts_{A,B} J= K st K_\Gamma =I K|_\delta=J A\lts B < C iff A<C [I] or B

\end{center} 
\caption{Dicograph sequent calculus with dicographs of atoms as sequents}
\label{sequentdicographs} 
\end{figure}

An induction on either sequent calculus given in this paper shows that: 

\begin{proposition} \label{proppartition} 
Let $\delta$ be a proof a dicograph sequent $R$, and let $\pi_\delta=(B,R)$ be the corresponding proof net. Then the axioms and atoms of $\pi_\delta$ can be partitioned 
into two classes $\Pi_1=(a_i \edge a_i^\perp)_{i\in I_1}$ and $\Pi_2=(a_i \edge a_i^\perp)_{i\in I_2}$ in such a way that either: 
\begin{enumerate} 
\item \label{partitionbef} there are only arcs from $\Pi_1$ to $\Pi_2$
\item \label{partitiontimes} the only edges between $\Pi_1$ and $\Pi_2$ are a $\tssp$ connection: calling  $R_1=R\restr_{\Pi_1}$ and $R_2=R\restr_{\Pi_2}$,  $R_1=A_1\pasp T_1$,  $R_1=A_2\pasp T_2$, and $R=(A_1\tssp A_2)\pasp T_1 \pasp T_2$
\end{enumerate} 
\end{proposition}

\begin{figure} 
\begin{center} 
\begin{tikzpicture}[PS,node distance=5mm and 5mm]
%  \draw[help lines] (-1,-5) grid (3,1) ;

  \node[vertex,label=$a$] (a) {} ;
  \node[vertex,label=$\llneg{a}$, left=of a] (acomp) {} ;

  \node[vertex,label=$\llneg{b}$, above right=of a] (bcomp) {} ;
  \node[vertex,label=$b$,right=1.2 of bcomp] (b) {} ;

  \node[vertex,label=left:$c$, below right=of a] (c) {} ;
  \node[vertex,label=right:$\llneg{c}$,right=of c] (ccomp) {} ;

  \node[vertex,label=right:$\llneg{d}$, below=of ccomp] (dcomp) {} ;
  \node[vertex,label=left:$d$,below=of c] (d) {} ;

  \node[vertex,label=below:$\llneg{f}$, below left=of d] (fcomp) {} ;
  \node[vertex,label=below:$f$, left=of fcomp] (f) {} ;

  \node[vertex,label=below:$\llneg{e}$, below right=of fcomp] (ecomp) {} ;
  \node[vertex,label=below:$e$,right=1.2 of ecomp] (e) {} ;

  \draw[axiom] (a) -- (acomp) ;
  \draw[axiom] (b) -- (bcomp) ;
  \draw[axiom] (c) -- (ccomp) ;
  \draw[axiom] (d) -- (dcomp) ;
  \draw[axiom] (e) -- (ecomp) ;
  \draw[axiom] (f) -- (fcomp) ;

  \llop{\rlink}{c}{bcomp}
  \llop{\cutlink}{a}{c,bcomp}
  \llop{\rlink}{ccomp}{dcomp}
  \llop{\rlink}{e}{b}
  \llop{\rlink}{ecomp}{fcomp}
  \llop{\cutlink}{d}{fcomp,ecomp}
  \llop{\rlink}{acomp}{f}
\end{tikzpicture}
\end{center} 
\caption{A proof net with no corresponding sequent calculus proof (found with Lutz Stra{\ss}burger)} 
\label{counterex} 
\end{figure} 

\pagebreak 

\begin{proposition} 
There does exist a proof net without any sequent calculus proof for example the one in Figure \ref{counterex}. 
\end{proposition} 

\begin{proof} 
First one as to observe that the proof structure in Figure \ref{counterex} is a proof net, i.e. contains no chordless alternate elementary circuit: indeed,  it contains no alternate elementary circuit. 

Because of Proposition \ref{proppartition}, there should  exists a partition into two parts with 
\begin{enumerate} 
\item  either only arcs from one part to the other part, 
\item or a tensor connection between the two parts. 
\end{enumerate} 

If the first case applies, i.e. if there were a partition into two parts with only arcs from one to another, all by vertices connected with an undirected edge, be it a $B$ or an $R$ edge, should be in the same component. 
$a,a^\perp,b,b^\perp,c,c,c^\perp$  should be in the same component say $\Pi_1$
 and 
  $f,f^\perp,d,d^\perp, e,e^\perp$ should be in the same component say $\Pi_2$, but this is impossible because there are both an $R$ arc from $\Pi_1$ to $\Pi_2$, e.g.  $a^\perp \fle f$, and an $R$ arc from $\Pi_2$ to $\Pi_1$, e.g. $e\fle b$. So the first case does not apply.

Because the first case does not apply,  there should exist two parts, with a  tensor rule as the only connection between two parts.
The two possible tensors are $a\tssp (c\befsp b^\perp)$ and $d\tssp (e^\perp \befsp f^\perp)$, but it is impossible: 
\begin{itemize} 
\item $a\tssp (c\befsp b^\perp)$ cannnot be the only connection between the two parts, as there exists an undirected path fro $c$ to $a$ not using any of the two tensor $R$ edges: $$c \arete{B} c^\perp \aretevers{R} d^\perp \arete{B} d \arete{B} f^\perp \arete{R} f \aretede{R}  a^\perp$$ 
\item $d\tssp (e^\perp \befsp f^\perp)$ cannnot be the only connection between the two parts, as there exists an undirected path from  $f^\perp$ to $d$ not using any of the two tensor $R$ edges: $$f^\perp \arete{B}  f  \aretede{R} a^\perp \arete{B} a \arete{R} c \arete{B} c^\perp \aretevers{R}  d^\perp \arete{B} d$$ 
\end{itemize} 
\end{proof}

In next Section \ref{PomsetMLLrewriting} we shall see that the correct proof net in Figure~\ref{counterex} 
can be derived from an axiom 
$
(a\pasp a^\perp) 
\tssp 
(b\pasp b^\perp) 
\tssp 
(c\pasp c^\perp) 
\tssp 
(d\pasp d^\perp) 
\tssp 
(e\pasp e^\perp) 
\tssp 
(f\pasp f^\perp) 
$ by means of the  rewriting rules  of Figure \ref{subdicograph} but $\ts\pa4$, see Figure \ref{counterexInDeepPomset}  and in SBV as well, see Figure  \ref{counterexInSBV}.

\section{Pomset logic in deep inference style} 
\label{PomsetMLLrewriting} 

In \cite{Ret98roma} (1998), I considered the rewriting rules of Figure~\ref{subdicograph} which preserves correctness (but $\ts\pa4$), but as Gugliemi noticed in \cite{Gug2007} (2007) I did not used the rewriting rules as a proof calculus to derive tautologies. However, in the conclusion of \cite{Ret99rr}  where I show that  the rewriting rules that are correct and concerns the MLL connectives ($\ts\pa3$ and $\ts\pa2=MIX$) are equivalent to MLL,  I explain that one could do the same for pomset logic with the rules of  Figure~\ref{subdicograph} that preserve correctness. This rewriting view was developed from 2001 with terms rather than graphs by Guglielmi and Stra{\ss}burger, as the calculus of structures   \cite{Gug99,GugStr01,Strassburger03phd}.

Before we define a rewriting deductive system for pomset logic, let us revisit (as we did in \cite{Ret99rr,Ret03tcs}) the deductive system of Multiplicative Linear Logic (MLL). 
Those results are highly inspired from proof nets, but once they are established they can be presented before proof nets are defined. 

In this section a sequent is simply a dicograph of \emph{atoms} which as explained above  can be viewed using \emph{folding} of Section \ref{foldunfoldsequent} as a dicograph of formulas or as an \sp\ order between formulas depending on how many \emph{folding} transformations and which one are performed. 

Regarding, multiplicative linear logic (MLL), observe that  ${\scriptscriptstyle{AX}}_n=\bigotimes_{1\leq i \leq n} (a_i\pa {a_i}^\perp)$ is the largest cograph  w.r.t. inclusion that can be derived in MLL with the $(a_i\pa {a_i}^\perp)$ as axioms: any additional R edge or arc would make a direct \ae-circuit with an axiom. However, ${\scriptscriptstyle{AX}}_n$ is actually derivable in  MLL, hence in any extension of MLL:

$$ 
\begin{prooftree} 
\begin{prooftree} %C 
\begin{prooftree} % B 
\begin{prooftree} %A
\begin{prooftree} %1
\seq a_1, {a_1}^\perp \qquad 
\justifies 
\seq {\scriptscriptstyle{AX}}_1: a_1 \pa {a_1}^\perp
\using \pa 
\end{prooftree} %1
\begin{prooftree} %2 
\seq a_2, {a_2}^\perp
\justifies 
\seq a_2 \pa {a_2}^\perp
\using \pa 
\end{prooftree} %2
\justifies 
\seq {\scriptscriptstyle{AX}}_2: \ts_{1\leq i \leq 2} ( a_i \pa {a_i}^\perp) 
\using \ts 
\end{prooftree} %A
\begin{prooftree} %3
\seq a_3, {a_3}^\perp
\justifies 
\seq a_3 \pa {a_3}^\perp
\using \pa 
\end{prooftree} %3 
\justifies 
\seq {\scriptscriptstyle{AX}}_3: \ts_{1\leq i \leq 3} ( a_i \pa {a_i}^\perp) \using \ts 
\end{prooftree} % B
\begin{prooftree} %4
\seq a_4, {a_4}^\perp
\justifies 
\seq a_4 \pa {a_4}^\perp
\using \pa 
\end{prooftree} %4 
\justifies 
 \qquad \qquad\qquad \seq{\scriptscriptstyle{AX}}_4: \ts_{1\leq i \leq 4} ( a_i \pa {a_i}^\perp) \qquad \qquad \qquad \qquad \qquad \cdots 
\using \ts 
\end{prooftree} % C
\justifies \seq {\scriptscriptstyle{AX}}_5: \cdots 
\using \ts 
\end{prooftree} 
$$

\subsection{Standard multiplicative linear logic as cograph rewriting} 
\label{mllrewrite} 

In \cite{Ret99rr},  we considered an alternative way to derive theorems of usual multiplicative linear logic MLL, by considering a formula as a binary relation, and more precisely as a cograph over its atoms, by viewing $\tssp$ as $\ts$ and $\pasp$ as $\pa$. As there is no $\bef$ connective in linear logic the series composition is not used, and there is no \sp\ order on conclusions.

Because of the chapeau of the present Section \ref{PomsetMLLrewriting}
 any  sequent of MLL can be viewed is a cograph\linebreak  $C[a_1,a_1^\perp,a_2, a_2^\perp,\ldots a_n,a_n^\perp]$ on $2n$ atoms that is included into $AX_n$. Because of Theorem \ref{subdicograph}, 
$AX_n$ 
rewrites to\linebreak  
$C[a_1,a_1^\perp,a_2, a_2^\perp,\ldots a_n,a_n^\perp]$ using the 
rules of Figure \ref{subdicograph} that concern $\pa$  and $\ts$ i.e. $\ts\pa4$, $\ts\pa3$
and $\ts\pa2$. 
Observe that when viewed as a linear implication  (considering the rules involving those two connectives), the first line $\tssp\pasp4$ is an incorrect linear implication, while $\ts\pa3$ is derivable in MLL and $\ts\pa2$ in MLL+MIX where the rule MIX is the one studied in \cite{FR94}, which also is derivable with $\ts\pa2$: 

$$
\begin{prooftree} 
\seq \Gamma \qquad \seq \Delta 
\justifies 
\seq \Gamma,\Delta 
\using \mathrm{MIX} 
\end{prooftree} 
$$ 

Actually \emph{all}  tautologies of multiplicative linear logic MLL can be derived using  $\ts\pa3$ from an axiom\linebreak  $AX_n=\bigotimes_{1\leq i\leq n} (a_i\pa {a_i}^\perp)$, and \emph{all}  tautologies   of linear logic enriched with the  MIX rule, MLL+MIX, can be derived by $\ts\pa3$ and $\ts\pa2$ (MIX). 

Thus, we can define a proof system gMLL for MLL working with sequents as cographs of atoms  as follows.  
Axioms are $AX_n$: $\tssp_i\  (a_i\pasp {a_i}^\perp)$ (the two dual atoms are connected by an edge in a different relation called $A$ for $A$ axioms). 
There is just one deduction rule presented as a rewrite rule (up to commutativity and associativity): $\ts\pa3$. 

Let us call this deductive system gMLL (g for graph), then  \cite{Ret99rr,Ret03tcs}   established that cograph rewriting is an alternative proof systems to MLL and MLL+MIX.

\begin{theorem} \label{rewMLL} 
MLL proves a sequent $\seq \Gamma$ with $2n$ atoms 
if and only if gMLL proves the unfolding $\Gamma^{cog}$ of $\Gamma$ (the cograph $\Gamma^{cog}$ of atoms corresponding to $\Gamma$, that is the $\pa$ of the unforging of each formula in $\Gamma$), i.e. $AX_n$ rewrites to $\Gamma^{cog}$ using $\ts\pa3$. 

MLL+MIX proves a sequent $\seq \Gamma$ with $2n$ atoms 
if and only if gMLL+mix proves the unfolding  $\Gamma^{cog}$ of $\Gamma$, i.e. $AX_n$ rewrites to $\Gamma^{cog}$ using $\ts\pa3$ and $\ts\pa2$. 
\end{theorem} 

\begin{proof} Easy induction on sequent calculus proofs see e.g. \cite{Ret99rr,Ret03tcs}. Stra{\ss}bruger made a direct proof in  \cite{Strassburger03phd}. \end{proof} 

 Thus all MLL proofs can be obtained that way from axioms, but despite Sta{\ss}burger result,  unfortunately for pomset logic,  it is hard to prove it directly on a non inductive notion of proof like proof nets.

\begin{proposition}
The calculi gMLL and gMLL+mix can safely be extended to structured sequents of formulas of MLL (not just atoms), i.e. cographs of MLL formulas  with the rules of folding and unfolding with the same results. 
\end{proposition}

\begin{proof} This is just an easy remark, based on proof nets, which can be viewed as a consequence Subsection \ref{PNling}. \end{proof}

\pagebreak

\subsection{Deep pomset is SBV} 
\label{sbv}

The above result for MLL suggests to present  pomset logic as a rewriting system from ``axioms" that are tensors of $x\pa x^\perp$ axioms,  with the rewriting rules of Figure \ref{subdicograph} except the one that does not preserve correctness, namely $\ts\pa4$ ---  the \sp-pomset sequent calculus of Figure \ref{sequentsporders} is defined along the same lines.  This rewriting calculus is rather natural because the rewriting rules preserve correctness (but $\ts\pa4$),  the rewriting rules themselves correspond to proofs  and to canonical linear maps in coherence spaces. 

This  suggests that a rewriting system defined as gMLL+mix in the previous section (but with dicographs instead of cographs) might yield all correct proof nets, but this is still an open question. Let us call \emph{deep pomset} the rewriting system for pomset logic defined by axioms and rewriting rules:

\begin{description} 
\item[Axioms]  $AX_n=\tssp_{1\leq i \leq n} (a_i\pasp {a_i}^\perp)$ is a tautology. 
\item[Rules]  Whenever a dicograph of atoms $D$ is  a tautology, so is the dicograph  $D'$ obtained  by any of the 10 rules $\ts\pa3$, $\ts\pa2$, $\ts\bef4$, $\ts\bef3l$, $\ts\bef3r$, $\ts\bef2$, 
$\bef\pa4$, $\bef\pa3l$, $\bef\pa3r$, $\bef\pa2$ of Figure \ref{subdicograph} --- i.e. all rules of Figure \ref{subdicograph} but $\ts\pa4$. Observe the $D'$ has the same atoms as $D$. 
\end{description} 

For simplicity, we leave out the circumflex accents which draw a distinctions between  the logical connectives ($\pa,\bef,\ts$) and the corresponding operations on dicographs ($\pasp,\befsp,\tssp$), because in this section there are only dicographs, denoted by terms. 

\begin{figure} 
$$
\begin{array}{r@{\ \rightsquigarrow\ }l}
\mbox{Axiom:}&
(e\pa e^\perp) 
\ts 
(b\pa b^\perp) 
\ts 
(c\pa c^\perp) 
\ts 
(f\pa f^\perp) 
\ts 
(a\pa a^\perp) 
\ts 
(d\pa d^\perp) 
\\  
\ts{\bef}2 & [(e^\perp  \pa e) \bef (b^\perp  \pa b)] \ts 
(c\pa c^\perp) 
\ts 
(f\pa f^\perp) 
\ts 
(a\pa a^\perp) 
\ts 
(d\pa d^\perp) 
\\ 
\bef\pa4 & [(e^\perp  \bef b^\perp ) \pa (e \bef b)] \ts 
(c\pa c^\perp) 
\ts 
(f\pa f^\perp) 
\ts 
(a\pa a^\perp) 
\ts 
(d\pa d^\perp) 
\\ 
\ts\pa3 & 
[\{(e^\perp  \bef b^\perp )  \ts 
(c\pa c^\perp) 
\ts 
(f\pa f^\perp)\} 
\pa (e \bef b)]
\ts 
(a\pa a^\perp) 
\ts 
(d\pa d^\perp) 
\\ 
2 \times \ts{\bef}2   & [\{((c\pa c^\perp ) \bef (e^\perp  \pa b^\perp )) \bef (f \pa f^\perp ))\} \pa (e \bef b) ]  
\ts 
(a\pa a^\perp) 
\ts 
(d\pa d^\perp) 
\\ 
2 \times {\bef}\pa4   &  [(c\bef b^\perp \bef f) \pa (c^\perp  \bef e^\perp  \bef f^\perp )  \pa (e \bef b) ] \ts 
(a\pa a^\perp) 
\ts 
(d\pa d^\perp)\\  

\ts\pa3  &  [\{(a\pa a^\perp) \ts (c\bef b^\perp \bef f) 
 \}\pa (c^\perp  \bef e^\perp  \bef f^\perp )  \pa (e \bef b) ] 
\ts 
(d\pa d^\perp)\\  

\ts{\bef}3l  &  [(\{(a\pa a^\perp) \ts (c\bef b^\perp)\} \bef f)  
\pa (c^\perp  \bef e^\perp  \bef f^\perp )  \pa (e \bef b) ] 
\ts 
(d\pa d^\perp)\\  

\ts\pa3  &  [(\{a^\perp \pa (a \ts (c\bef b^\perp))\} \bef f) 
\pa (c^\perp  \bef e^\perp  \bef f^\perp )  \pa (e \bef b) ] 
\ts 
(d\pa d^\perp)\\

\bef\pa3  &  [ (a \ts (c\bef b^\perp)) \pa (a^\perp \pa \bef f) 
\pa  (c^\perp  \bef e^\perp  \bef f^\perp )  \pa (e \bef b)]
\ts 
(d\pa d^\perp)
\\  
\ts\pa3 &   (a \ts (c\bef b^\perp)) \pa (a^\perp \bef f) 
\pa \{ (c^\perp  \bef e^\perp  \bef f^\perp ) \ts 
(d\pa d^\perp)\} \pa (e \bef b)
\\

\ts\pa3  &  (a \ts (c\bef b^\perp)) \pa (a^\perp \bef f)  
\pa  (c^\perp  \bef \{ [ (e^\perp  \bef f^\perp ) \ts 
d] \pa d^\perp\} \pa (e \bef b)
\\  
{\bef}\pa3r  &  (a \ts (c\bef b^\perp)) \pa (a^\perp \bef f)
\pa   ((e^\perp  \bef f^\perp ) \ts d) \pa (c^\perp  \bef d^\perp) \pa (e \bef b)  
\end{array} 
$$
\caption{The derivation of the proof net of Figure \ref{counterex} in Deep Pomset.} 
\label{counterexInDeepPomset} 
\end{figure}

\begin{figure} 
\label{SBVrules}
\begin{itemize}
    \item 
The rewriting rules $\ts\pa3$, $\ts\pa2$, $\ts\bef4$, $\ts\bef3l$, $\ts\bef3r$, $\ts\bef2$, 
$\bef\pa4$, $\bef\pa3l$, $\bef\pa3r$, $\bef\pa2$ (but $\ts\pa4$) of  Figure \ref{subdicograph} are the structural rules of  SBV.  Because of $\un$  the interchange law on four terms is enough in SBV while Pomset Logic (without $\un$) requires the ternary and binary rules as well.
\item 
The rule $\un{\downarrow}$ says $\un$ may appear of vanish "anywhere": 

\centerline{ 
$
\begin{array}{rcl} 
\multicolumn{3}{c}{\mbox{Rule $\un{\downarrow}$}}\\ 
S[T]&\leftrightsquigarrow&S[(T\pa\un)]\\   S[T]&\leftrightsquigarrow&S[(T\bef\un)]\\  S[T]&\leftrightsquigarrow&S[(\un\bef T)]\\   S[T]&\leftrightsquigarrow&S[(T\ts\un)]\\ 
\end{array}
$}  

\item 
The rule rule $a{\downarrow}$ says that $\un$ may be replaced with an ``axiom" $a_i\pa a_i^\perp$: 

\centerline{$\begin{array}{rcl}
S[\un]&\rightsquigarrow&S[(a_i\pa a_i^\perp)]\\   
\end{array}
$} 

\item The rule $a{\uparrow}$ says that 
a pair of edge equivalent dual atoms (i.e. equivalent and connected with a tensor) can vanish every where (in the original formulation of SBV $S[(a_i\ts a_i^\perp)]$ rewrites into $S[\un]$ but given the above rules for $\un$ above it is unnecessary). T

\centerline{$\begin{array}{rcl}
S[(a_i\ts a_i^\perp)]&\rightsquigarrow&S\restr_{x\not\in\{a_i,a_i^\perp\}}\\
\end{array}
$} 
\end{itemize} 

    \caption{The term calculus SBV \cite{GugStr01}, written with MLL/pomset symbols. Observe that associativity, commutativity are implicit in handsome proof nets (graphs) but invertible rewriting rules in SBV.} 
    \label{sbvfig}
\end{figure}

\begin{figure} 
$$
\begin{array}{r@{\ \rightsquigarrow\ }l}
\mbox{Axiom} &\un \\ 
a{\downarrow} & (e^\perp  \pa e) \\ 
\un a{\downarrow} & (e^\perp  \pa e) \ts (b^\perp  \pa b) \\ 
\ts\bef2 & (e^\perp  \pa e) \bef (b^\perp  \pa b) \\ 
\bef\pa4 & (e^\perp  \bef b^\perp ) \pa (e \bef b)  \\ 
(\un a{\downarrow})\times 2  &  ((c\pa c^\perp ) \ts (e^\perp  \pa b^\perp ) \ts (f \pa f^\perp )) \pa (e \bef b) \\ 
\ts\bef2 x 2   & ((c\pa c^\perp ) \bef (e^\perp  \pa b^\perp ) \bef (f \pa f^\perp )) \pa (e \bef b)  \\ 
\bef\pa4  x 2 &  (c\bef b^\perp \bef f) \pa (c^\perp  \bef e^\perp  \bef f^\perp )  \pa (e \bef b)  \\ 
\un a{\downarrow}  & (((c\bef b^\perp )\ts (a \pa a^\perp ))\bef f) \pa (c^\perp  \bef e^\perp  \bef f^\perp )  \pa (e \bef b)  \\ 
\ts\pa3 &  (((c\bef b^\perp )\ts a) \pa a^\perp )\bef f) \pa (c^\perp  \bef e^\perp  \bef f^\perp )  \pa (e \bef b)  \\ 
\bef\pa3 &  ((c\bef b^\perp )\ts a) \pa (a^\perp \bef f) \pa (c^\perp  \bef e^\perp  \bef f^\perp )  \pa (e \bef b)  \\ 
\un a{\downarrow} &  ((c\bef b^\perp )\ts a) \pa (a^\perp \bef f) 
\pa (c^\perp  \bef ((e^\perp  \bef f^\perp ) \ts (d \pa d^\perp ))  
\pa (e \bef b)  \\ 
\ts\pa3 & 
((c\bef b^\perp )\ts a) \pa (a^\perp \bef f) 
\pa \{c^\perp  \bef [((e^\perp  \bef f^\perp) \ts d) \pa d^\perp]\}  
\pa (e \bef b)  \\ 
\bef\pa3
 & ((c\bef b^\perp )\ts a) \pa (a^\perp \bef f)  \pa (((e^\perp  \bef f^\perp ) \ts d) \pa (c \bef d^\perp)  \pa (e \bef b) 
\end{array}
$$
\caption{An example of an SBV derivation: the dicog proof net of Figure \ref{counterex} in SBV (thanks to Lutz Stra{\ss}burger). We grouped the rules $\un{\downarrow}$ and $a{\downarrow}$. As we proved in this paper derivation can be converted into a Deep Pomset derivation, the one in Figure~\ref{counterexInDeepPomset}} 
\label{counterexInSBV} 
\end{figure} 

However, what is simple is to show that pomset logic as graph rewriting i.e. deep pomet is equivalent to the SBV rewriting system defined in Figure \ref{sbvfig} --- we follow the simple presentation given in \cite{GugStr01}.  The SBV system is defined as term rewriting, rather than as dicographs rewriting. In SBV term system there are bidirectional rewriting rules that are simply equality of the dicographs (e.g. associativity or commutativity of the order operations). Another difference is that axiom is expressed in BV as $\circ$, that I shall write $\un$\footnote{The $\circ$ symbol is cute, but $\circ$ is unusual for denoting a unit, or a truth value, or the set $\{*\}$ with a single element which actually is the multiplicative unit in the category of coherence spaces, so I prefer the standard notation 
``$\un$".}, a unit for all the three connectives,
which may appear anywhere, and which may be rewritten as $a\pa a^\perp$. 

\pagebreak 

\subsubsection[Simulating  the atomic-cut reduction in Deep Pomset]{Simulating $a{\uparrow}$, the atomic-cut reduction in Deep Pomset} 

In the proof net framework,  $a{\uparrow}$ is just an atomic cut i.e. a formula $a \simts a^\perp$ that is not a subformula: the dicograph is  $t\pa(a \simts a^\perp)$. Let us first see that this can be done within pomset logic.

\begin{proposition}\label{extractatsaperp}  Let $S[t*u]$ be a dicograph. Then $S[t*u]$  with $*=\bef,\ts,\pa$
rewrites into $S[t]\pa u$ with the correct  rewriting rules only, i.e. $\ts\pa4$ can be avoided. 
\end{proposition} 

\begin{proof} Because $S[t*u]\subset S[t]\pa u$ here is no doubt the largest dicograph $S[t*u]$ rewrites into the smaller one $S[t]\pa u$, but we need to check that $\ts\pa4$ is unnecessary. 
This is rather straight forward, by induction on the structure of $S[]$. 
%If $S=t*u$ then the result holds with $\ts\pa2$ of $\bef\pa2$. If $S[t*u]=S_1[t*u]\star S_2$ with with $*=\bef,\ts,\pa$ (resp. $S_2\bef S_1[t*u]$, because $\bef$ is not commutative, it is a separate case), by induction hypothesis $S[t*u]=S_1[t*u]\star S_2$ rewrites to $(S_1[t]\pa u)\star S_2$ (resp. $S_2\bef (S_1[t] \pa u)$) which finally rewrites into $(S_1[t]\star S_2) \pa u=S[t]\pa u$ using $\ts\pa3$ or $\bef\pa3l$ or $\pa$ associativity  (resp.  rewrites into $(S_2[t]\bef S_1)\pa u=S[t]\pa u$ using $\bef\pa3r$). So $S[t*u]$  with $*=\bef,\ts,\pa$ rewrites into $S[t]\pa u$ without using $\ts\pa4$.
\end{proof}

\begin{definition}\label{atomcutredok} 
Given a proof net $\pi=(V,B,R)$ with $R=S[u*(a\ts a^\perp)]$. An \emph{atomic cut reduction} is a transformation of $\pi$ into $\pi'=(V',B',R')$,
with $V'=\{a,a^\perp\}$, $B'=B\setminus \{B(a)\frac{B}{\qquad} a, a^\perp \frac{B}{\qquad} B(a^\perp)\} \cup \{B(a)\frac{B}{\qquad}B(a^\perp)\}$ and $R'=S[u]$ --- the atoms $a$ and $a^\perp$ are not anymore in the domain of the R dicograph.  The expression $B(x)$ denotes the unique vertex $B$-related to $x$.
\end{definition} 

\begin{proposition}\label{atomiccutelim}
When an atomic cut reduction is performed on a proof net $\pi=(V,B,R)$ with $R=S[u*(a_i\ts a_i^\perp)]$ yielding $\pi'=(V',B',R'=S[u])$ as in the above definition $\pi'$ is a correct proof net as well. 
\end{proposition}

\begin{proof} 
We view atomic cut reduction as a two step process. The first step consists in turning $\pi$ into $\pi^-=(V,B,S[t]\pa (a\ts a^\perp)$ and as shown in Proposition \ref{extractatsaperp} this transformation preserves correctness ($\ts\pa4$ is not used): $\pi^-$ is correct. The second step consists in replacing in $\pi^-$ the  sequence of three edges 
$\{B(a)\frac{B}{\qquad} a, 
a^\perp \frac{R}{\qquad} B(a^\perp)\} \cup \{B(a)\frac{B}{\qquad}B(a^\perp)\}$ with a single B edge $B(a)\frac{B}{\qquad}B(a^\perp)$. 

Observe that in $\pi^-$ there are no R edges incident to $a$ nor to $a^\perp$. 
Hence the  \ae circuits are unchanged in $\pi^-$ and i $\pi'$ are the same up to the replacement of $\{B(a)\frac{B}{\qquad} a, 
a^\perp \frac{R}{\qquad} B(a^\perp)\} \cup \{B(a)\frac{B}{\qquad}B(a^\perp)\}$ with the single B edge $B(a)\frac{B}{\qquad}B(a^\perp)$. 
If one \ae circuit of $\pi'$ would be chordless so would be its image in $\pi^-$
\end{proof} 

Next proposition shows that if a proof net is derivable in deep pomset, so is the proof net obtained by performing an atomic cut reduction: 
\begin{proposition}\label{reachfromaxioms} 
Let $\pi=(V,B,R)$ with $R=S[t*(a\ts a^\perp)]$ be a proof net derivable in deep pomset and let $\pi'=(V,B,R[t])$ be the proof net obtained from $\pi$ by an atomic cut reduction as in the above definition. If $Ax_{n+1}$ rewrites to $\pi$ with only the correct rewriting rules, then $Ax_n$ (one axiom less) rewrites into $\pi'$ with only the correct rewriting rules. 
\end{proposition} 

\begin{proof} 
Let ${\underline{a}}^\perp=B(a)$ and ${\underline{a}}=B(a^\perp)$ the atoms that are linked to $a$ and $a^\perp$ with axioms, that are the two end-vertices of the same B-edge in a proof net (or the ones that are created jointly by a $a{\downarrow}$ in SBV).

Let us call $\delta$ a  sequence of rewritings leading from $Ax_{n+1}=({\underline{a}}^\perp  \pa a)\otimes 
(a^\perp \pa {\underline{a}})\otimes (\otimes_i (a_i\pa {a_i}^\perp))$ to  $R=S[t*(a\ts a^\perp)]$. 

Observe that the restriction of $R$ to ${\underline{a}}^\perp,a, a^\perp,{\underline{a}}$ can only be $(a\ts a^\perp)\pa {\underline{a}} \pa {\underline{a}}^\perp$, because $a\simts a^\perp$ in $R$, and because $R$ is a subdicograph of $Ax_{n+1}$. Thus the rewriting $\delta$ includes at some point the rewriting the sub term $w=({\underline{a}}^\perp  \pa a)\otimes 
(a^\perp \pa {\underline{a}})$ by an occurrence $i$ of the rewriting rule $\ts\pa3$ either into $u={\underline{a}}^\perp  \pa (a \otimes 
(a^\perp \pa {\underline{a}}))$  or into $v=(({\underline{a}}^\perp  \pa a)\otimes 
a^\perp) \pa {\underline{a}})$ --- that are the only possibilities since at the end of the rewriting $\delta$ we have $a\simts a^\perp$ i.e. $a$ and $a^\perp$ should be kept linked by a $\tssp$. 

Because at the end of the rewriting $\delta$ we have $a\simts a^\perp$, if $\delta$  reduced $w$ to $u$ (resp. to $v$), then $\delta$ reduces later on the sub term $u'=(a \otimes 
(a^\perp \pa {\underline{a}}))$ of $u$ (resp. the sub term $v'=
({\underline{a}}^\perp  \pa a)\otimes  a^\perp$  of $v$) 
into $u''= (a \otimes 
a^\perp) \pa {\underline{a}}$ (resp. into $v''=
{\underline{a}}^\perp  \pa (a\otimes  a^\perp)$) by an occurrence $j$ the rewriting rule $\ts\pa3$.

Because the  $\ts\pa3$ rules $i$ and $j$ commute with the rewriting rules that precede them, we may assume 
that $\delta$ starts with the  $\ts\pa3$ rules $i$ and $j$ turning 
$({\underline{a}}^\perp  \pa a)\otimes  (a^\perp \pa {\underline{a}})$ of $Ax_{n+1}$ 
into $(a\ts a^\perp) \pa ({\underline{a}} \pa {\underline{a}}^\perp)$, then followed by a rewriting $\delta^-$ from $[(a\ts a^\perp) \pa ({\underline{a}} \pa {\underline{a}}^\perp)]  
\otimes (\otimes_i (a_i\pa {a_i}^\perp))$
to  $R=S[t*(a\ts a^\perp)]$. 

The result is obtained by considering the projection $\delta^-\restr_{x\not\in\{a,a^\perp}$ of $\delta^-$  on the dicographs from which  $a$ and $a^\perp$ are dropped out. The only prohibited rule $\ts\bef4$ cannot result obtained from a correct rule by dropping out an atom. The correct rewriting derivation $\delta^-\restr_{x\not\in\{a,a^\perp}$ reduces $Ax_n=({\underline{a}} \pa {\underline{a}}^\perp)  
\otimes (\otimes_i (a_i\pa {a_i}^\perp))$ into $R=S[t]$. 
\end{proof}

\subsubsection{Dealing with the unit in Deep Pomset} 

To establish the correspondence between deep pomset and SBV either we can enrich pomset logic with a unit, or rather consider the dicographs proved by SBV without any unit remaining.

We encode $\un$ in RnB graphs with an axiom, i.e. a B edge whose two ends $\epsilon$ and $\epsilon^\perp$ are always par-equivalent 
$\epsilon \simpa \epsilon^\perp$
in the dicograph $t$ i.e. $t$ is $t[\epsilon\pa\epsilon^\perp]$ --- as the derivation proceeds, $\epsilon$ and $\epsilon^\perp$ are never driven apart: in other words we define $\un=(\epsilon\pa\epsilon^\perp)$ --- and we may use $\un$ or $(\epsilon\pa\epsilon^\perp)$ to denote it in dicographs.\footnote{There are other faithful encodings of $\un$ but this one 
 preserves the structure of the proof net: a simple graph, whose 
B edges define a perfect matching, R edges are a dicograph and every alternate elementary cycle has an R chord.}

A case study of possible \ae-circuit and their R chords show the two following propositions: 
\begin{proposition}\label{adding1} 
Let $u$ and $t[x]$ be dicographs ($x$ stands for a vertex). Assume that $(B, t[u])$
is a correct proof net (every \ae-circuit contains an R chord). 
Then $(B\cup \epsilon B \epsilon^\perp , t[u\ts\un])$ is a correct proof net as well. 
\end{proposition}

\begin{proposition}\label{removing1}  
Let $u$ and $t[x]$ be dicographs ($x$ stands for a vertex). 
Assume that $B\cup \epsilon\--\epsilon^\perp , t[u\ts\un]$
is a proof net (every \ae-circuit contains an R chord). 
Then $B, t[u]$ is a correct proof net as well. 
\end{proposition}

\begin{proposition}\label{internalise} Given three dicographs $t,u,v$, 
$t[u]\otimes v$ rewrites to $t[u\ts v]$ and to $t[u\pa v]$ and $t[u\bef v]$ and $t[v\bef u]$.  
\end{proposition}

\begin{proof} 
Easy induction on the structure of dicograph with a ``hole" $t[]$. 
%$t[u]\otimes v$ rewrites to $t[u\ts v]$ by a straightforward induction on the structure of $t$. For instance, if $t[u]=t'[u]\bef t''$ then $t[u]\ts v=(t'[u]\bef t'')\ts v$  rewrites to $(t'[u]\ts v)\bef t''$ by $\ts\bef3g$; by induction hypothesis ($t'$ is smaller than $t$) $(t'[u]\ts v)$ rewrites to $(t'[u \ts v]$, so $t[u]\ts v=(t'[u]\bef t'')\ts v$ rewrites to $(t'[u]\ts v)\ts t''=t[u\ts v]$. The other cases e.g. $t[u]=t' \pa t''[u]$ etc are absolutely similar. 
%Then $t[v\bef u]$ rewrites to $t[u\pa v]$, $t[u\bef v]$ or  $t[v\bef u]$ using $\ts\pa2,\ts\bef2r,\ts\bef2l$. 
\end{proof}

In order to ease the correspondence with SBV, let us slightly extend deep pomset into \textbf{unitary deep pomset}  which involves  $\un$: 
\begin{description} 
\item[Axioms]  $AX_n=\ts_{1\leq i \leq n} (a_i\pa {a_i}^\perp)\ts \ts_{1\leq j \leq n} (\epsilon_j\pa {\epsilon_j}^\perp)$ is a tautology, and $\epsilon_j\neq \epsilon_{j'}$ whenever $j\neq j'$. 
\item[Rules]  Whenever a dicograph of atoms $D$ which is  a tautology rewrites to a dicograph $D'$ (hence with the sames atoms) by any of the 10 rules $\ts\pa3$, $\ts\pa2$, $\ts\bef4$, $\ts\bef3l$, $\ts\bef3r$, $\ts\bef2$, 
$\bef\pa4$, $\bef\pa3l$, $\bef\pa3r$, $\bef\pa2$ of Figure \ref{subdicograph} --- i.e. all rules of Figure \ref{subdicograph} but $\ts\pa4$. \emph{Whenever those rules are applied, they must never tear apart an $\epsilon_j$ from the related  $\epsilon_j^\perp$.} 
\item[Add1 / Rm1 / Subst1] The rules Add1 (insert $\un$) and Rm1 (remove $\un$) and Subst1 (replace $\un$ with $(a_i\pa a_i^\perp$) are respectively defined  as $\un,a{\downarrow},a{\uparrow}$ in SBV. Because of Propositions \ref{adding1} ($\un$) and  \ref{removing1} ($a{\uparrow}$) and common sense ($a\downarrow$) those rules preserve the correctness criterion. 

\begin{verse} 
Add1:  $t[u] \mapsto t[u\ts\un]$ (and an axiom $(\epsilon_j B {\epsilon_j}^\perp)$ with a fresh $j$ corresponding to $\un$ is added) 

Rm1: $t[u\ts\un] \mapsto t[u]$ (and the  axiom $(\epsilon_j B {\epsilon_j}^\perp)$ corresponding to $\un$ is removed) 

Subst1: Consists in replacing dual atoms $\epsilon_j$ and $\epsilon_j^\perp$ by dual atoms $a_i$ and $a_i^\perp$ (no matter whether $a_i$ replaces $\epsilon_j$ and $a_i^\perp$ replaces $\epsilon_j^\perp$ or the converse. 
\end{verse} 
\end{description}

\subsubsection{Result: SBV is Deep Pomset}

\begin{theorem}\label{sbv&pomset} 
Any SBV derivation yielding a dicograph without 1 can be mapped rule by rule into a unitary deep pomset derivation and vice versa. 
\end{theorem}

\begin{proof} There is an obvious one-to-one correspondence bewteen the rules of SBV and the ones of unitary deep pomset -- this is the reason why we added the rules Add1 Rm1 and Subst1 which respectively correspond to $\un,a{\uparrow},a{\downarrow}$

A little difference is that Subst1 acts upon the whole previous (rewriting) derivation while $s{\downarrow}$ is local.  
\end{proof} 

\begin{theorem}\label{sbv&pomset1} 
Any SBV derivation yielding a term without $\un$ 
can be mapped rule by rule into a deep pomset derivation (without the rules for $\un$ of unitary deep pomset derivations) and vice versa. 
\end{theorem} 

\begin{proof} 
Clearly, Deep Pomset is a subcalculus of SBV, rule by rule. 

So let us focus on the other direction, from SBV to Deep Pomset. 

The result is a consequence of  both Proposition 
\ref{reachfromaxioms} which states that the $a{\uparrow}$ can be simulated in Deep Pomset logic (without $\un$ rules) and of 
Proposition \ref{internalise} which allows to moves the $\un$s inside the dicograph, at the place where they are actually used by some rewrite rule. 

We turn any SBV derivation without any $\un$ in its final dicograph  into an SBV derivation which starts by a tensor of $\un$s that are immediately replaced with  axioms $a_i\pa a_i^\perp$ and without any $\un$-rules inside.  Observe that such an SBV derivation actually is a deep pomset derivation: simply erase  the initial tensor of  $\un$s and resume the derivation when it has been replaced with a tensor of axioms $a_i\pa a_i^\perp$. 

Observe that in an SBV proof the $\un$s are mandatory to start with. 
Given that they vanish during the derivation, each $\un$ either is deleted according to the equalities that say it is a unit for $\ts,\pa,\bef$ or is turned into $a_i\pa a_i^\perp$. 

Because during the derivation nothing is duplicated nor contracted, so atoms can be tracked we can erase the first kind of $\un$ from the derivation from where it appeared to where is disappeared, and this is an SBV derivation as well. 

From now on we may assume without lost of generality that the SBV derivation contains several 
$\un$ that are introduced in the derivation and later on replaced with a $a_i\pa a_i^\perp$ axiom. We proceed by induction on the number of such $\un$. Consider the first $\un$
that appears in the SBV derivation $d_{sbv}$, say at the n-th step of $d_{sbv}$, and disappears at step $n+k$ of the SBV derivation.

Let us call $d_1$ the SBV derivation this $\un$ appears, $d_2$ the part of the SBV derivation till the replacement of this $\un$ by $a_i\pa a_i^\perp$, and $d_3$ the end of this SBV derivation.

The new derivation without this $\un$ that appears in the middle of the derivation before being replaced with  $\ts(a_i\pa a_i^\perp)$ is defined as follows: 

\begin{enumerate}
    \item Let us add a $\un$ to the axiom of the SBV derivation, and immediately after, let us replace this $\un$ with $a_i\pa a_i^\perp$ which is going to replace $\un$ at the end of $d^1$. 
    \item The first part of the derivation consists in $d'_1$, that is $d_1$ with $\ts (a_i\pa a_i^\perp)$ at the end of every dicograph in the derivation. 
    \item Thereafter the derivation consist in moving the axiom to the place $(a_i\pa a_i^\perp)$ by deep pomset rules as in Proposition \ref{internalise}. 
    \item Then the derivation is $d'_2$ that is $d_2$ with $\un$ replaced with $(a_i\pa a_i^\perp)$. 
    \item The last part of the derivation is unchanged, it is $d_3$. 
\end{enumerate}

Hence we can turn the SBV derivation into an SBV derivation that starts with an axioms that is $\ts_{n\in I} \un$ which is immediately replaced with axioms $\ts_{n\in I} (a_i\pa a_i^\perp)$, yielding the same dicograph and only using deep pomset. 
\end{proof}

\begin{interrogation} 
We may wonder whether all proof nets 
can be obtained from $AX_n=\ts_{i\in I} (a_i\pa a_i^\perp)$ using only the correct rewriting rules (which simply are the correct inclusion patterns) of Figure~\ref{subdicograph} (all of them but $\ts\pa4$). This would provide an inductive definition of the proofs of pomset logic, different from the inductive definition provided by a sequent calculus with a sequentialisation theorem.  
\end{interrogation}

\subsection{Cut elimination in Deep Pomset and in  SBV} 
\label{cuts}

What about the cut rule? For such logical systems based on rewriting systems like gMLL(+MIX), of the \dicogr\ view of pomset logic,   which does not work with  "logical rules" in the standard sense, there are no binary rules that would combine a $K$ and a $K^\perp$. 

Hence, as we said earlier,  the natural view of a  cut is simply a tensor $K\ts K^\perp$ which in pomset logic never is inserted inside a $\ts$ formula, while SBV allows a generalisation with cuts $K\ts K^\perp$ occurring "inside" dicographs. 
The rule $i{\uparrow}$ of SBV  generalises $a\uparrow$: $i{\uparrow}$ rewrites any subterm $K\ts K^\perp$ to $\un$--- or suppresses $K\ts K^\perp$ if one does not want units.  A way to express cut-elimination in rewriting system is to say that the dicographs that can be derived with $i{\uparrow}$   can be derived without $i{\uparrow}$.

Corollary  \ref{cutasrew} shows that decomposing a cut $(A\otimes B)\ts (A^\perp \pa B^\perp)$ or 
$(A\bef B)\ts (A^\perp \bef B^\perp)$
into two smaller cuts $A\ts A^\perp$ and $B\ts B^\perp$ can be done by correct rewritings --- no matter whether this cut appears inside a dicograph. In Proposition \ref{reachfromaxioms}, we proved that whenever there exists a derivation of $S[t*(a\pa a^\perp)]$ ($a$ atomic) from  $Ax_n$, there is one derivation of $S[t]$ from $Ax_{n-1}$. From those two propositions, we obtain for free an alternative proof of cut elimination for SBV or for Deep Pomset Logic, a result that first appeared in \cite{Strassburger03phd}: 

\begin{theorem}[Cut elimination for SBV, \cite{Strassburger03phd}] 
Let $t$ be a dicograph derived with $i{\uparrow}$ from an axiom $Ax_m$. Then $t$ can be derived from an axiom $Ax_n$ with $n\leq m$ without the rule $i{\uparrow}$. 
\end{theorem}

As can be foreseen in those lines, Deep Inference set the stage for new conceptions of cut and of cut elimination, 
studied 
by Guglielmi and Stra{\ss}burger; presenting their work on Deep Inference falls beyond the scope of  this chapter on pomset logic.

\section{Grammatical use} 
\label{grammar} 

Relations like dicographs have pleasant  algebraic properties but when it comes to combining trees as in grammatical derivations, it is better to view the trees in order to have some intuition. So we first present proof nets with links before defining a grammatical formalism. 

\subsection{Proof nets with links} 
\label{PNling} 

In order to define a grammar of pomset proof nets, if is easier to use proof nets with links (the links have been presented in Figure~\ref{rnblinks} which look like standard proof nets: the formula trees of the conclusions $T_1,\ldots, T_n$ with binary connectives ($\pa,\ts,\bef$) and axioms linking dual atoms, together with an \sp\ partial order on the conclusions $T_1,\ldots, T_n$. 

It is quite easy to turn a \dicogpn~ proof net into a pomset proof net  using folding of Subsection \ref{foldunfoldsequent} --- and vice-versa using unfolding.  
A  dicograph proof net $\pi=(B,R)$ with R being $S[T_1,\ldots,T_n]$ with $S$ containing no $\tssp$ symbol --- $S$ is an \sp\ order --- corresponds to a pomset proof net $\pi^\sp$ with conclusions $T^f_1,\ldots,T^f_n$ where $T^f_i$ is the formula corresponding to $T_i$ obtained by replacing an operation on dicograph 
$\widehat *$ with the corresponding multiplicative connective $*\in\{\ts,\bef,\pa\}$. 
There usually are many ways to write a dicograph $R$ as a term $S[T_1,\ldots,T_n]$  depending on the associativity of $\ts,\pa,\bef$, commutativity  of $\ts,\pa$ and the $n$ may vary when the outer most $\pasp$ and $\befsp$ are turned into $\pa$ and $\bef$ connectives or not (as it is the case for $\pa$ in usual proof nets for MLL). In case the outer most connective of $R$ is $\tssp$,  $\pi$ necessarily has a single conclusion, $R=T_1$, and $S$ is the trivial \sp\ order on one formula.

The transformation from $\pi$ to $\pi^\sp$ can be done ``little by little" by allowing ``intermediate" proof structures whose conclusion is a dicograph of \emph{formulas}. Such a proof structure is said to be correct whenever every \ae\ circuit contains a chord, the formula trees being bicoloured as in Figure \ref{rnblinks} --- in figure \ref{foldingStepByStep} $\pi_1$ is the \dicogpn\ proof net, while $\pi_4$ is a pomset proof net with links having a single conclusion.

\tikzset{vertex/.append style={inner sep=2pt}}

\begin{figure} 
\centering
\subfigure[{$\pi_1$}]{%
\begin{forest} proof structure, for tree ={grow''=105,l=0pt,l sep=4pt}, commutative,
  erase level=1
  [,
  [,%precedes
  [,name=a par acomp times c par ccomp,for tree={grow''=195},erase level=3
    [,%times
      [,for tree={grow''=-90}, name=a par acomp,erase level=5
      [,s sep=15mm [,name=a,label=-90:$\alpha$,fill=black]  [,name=acomp,label=-90:$\llneg{\alpha}$,fill=black] ]
      ]
      [,for tree={grow''=120}, name=c par ccomp,erase level=5
      [,s sep=15mm [,name=c,label=$\gamma$,fill=black] [,name=ccomp,label=$\llneg{\gamma}$,fill=black] ]
      ]
    ]
    ]
    [,for tree={grow''=60}, name=b par bcomp,erase level=3
  [,s sep=15mm [,name=bcomp,label=$\llneg{\beta}$,fill=black] [,name=b,label=$\beta$,fill=black]]]
    ]
    ]
    \draw[Blink] (a) to (acomp) ;
    \draw[Blink] (b) to (bcomp) ;
    \draw[Blink] (c) to (ccomp) ;
    \llop{\cutlink}{a,acomp}{c,ccomp}
    \llop{\rlink}{a,acomp,c,ccomp}{b,bcomp}
\end{forest}}
\\ 
%\hspace{2cm}
\subfigure[{$\pi_2$}]{%
\begin{forest} proof structure, for tree ={grow''=105,l=0pt,l sep=4pt}, commutative,
  erase level=1
  [,
  [,%precedes
  [,name=a par acomp times c par ccomp,for tree={grow''=195},erase level=3
    [,%times
      [,for tree={grow''=-90}, name=a par acomp,erase level=5
      [,s sep=15mm [,name=a,label=-90:$\alpha$,fill=black]  [,name=acomp,label=-90:$\llneg{\alpha}$,fill=black] ]
      ]
      [,for tree={grow''=120}, name=c par ccomp,erase level=5
      [,s sep=15mm [,name=c,label=$\gamma$,fill=black] [,name=ccomp,label=$\llneg{\gamma}$,fill=black] ]
      ]
    ]
    ]
    [,for tree={grow''=60}, name=b par bcomp,label={[font=\scriptsize]0:$\beta \llpar \llneg{\beta}$},fill=black
  [,par,s sep=15mm [,name=b,label=$\beta$] [,name=bcomp,label=$\llneg{\beta}$]]]
    ]
    ]
    \draw[Blink] (a) to (acomp) ;
    \draw[Blink] (b) to (bcomp) ;
    \draw[Blink] (c) to (ccomp) ;
    \llop{\cutlink}{a,acomp}{c,ccomp}
    \llop{\rlink}{a,acomp,c,ccomp}{b par bcomp}
\end{forest}
}\\
\subfigure[{$\pi_3$}]{%
\begin{forest} proof structure, for tree ={grow''=105,l=0pt,l sep=4pt}, commutative,  erase level=1
  [,
  [,%precedes
  [,name=a par acomp times c par ccomp,for tree={grow''=195},erase level=3
    [,%times
      [,for tree={grow''=-90}, name=a par acomp,label={[font=\scriptsize]180:$\alpha \llpar \llneg{\alpha}$},fill=black
      [,par,s sep=15mm [,name=a,label=-90:$\alpha$]  [,name=acomp,label=-90:$\llneg{\alpha}$] ]
      ]
      [,for tree={grow''=120}, name=c par ccomp,label={[font=\scriptsize]180:$\gamma \llpar \llneg{\gamma}$},fill=black
      [,par,s sep=15mm [,name=c,label=$\gamma$] [,name=ccomp,label=$\llneg{\gamma}$] ]
      ]
    ]
    ]
    [,for tree={grow''=60}, name=b par bcomp,label={[font=\scriptsize]0:$\beta \llpar \llneg{\beta}$},fill=black
  [,par,s sep=15mm [,name=b,label=$\beta$] [,name=bcomp,label=$\llneg{\beta}$]]]
    ]
    ]
    \draw[Blink] (a) to (acomp) ;
    \draw[Blink] (b) to (bcomp) ;
    \draw[Blink] (c) to (ccomp) ;
    \llop{\cutlink}{a par acomp}{c par ccomp}
    \llop{\rlink}{a par acomp,c par ccomp}{b par bcomp}
\end{forest}
}
\\ 
%\hspace{2cm}
\subfigure[{$\pi_4$}]{%
\begin{forest} proof structure, for tree ={grow''=105,l=0pt,l sep=4pt}, commutative
  [,label={[font=\scriptsize]-100:$((\alpha \llpar \llneg{\alpha})\llts (\gamma \llpar \llneg{\gamma})) \llprecedes (\beta \llpar \llneg{\beta})$},fill=black
  [,precedes
    [,for tree={grow''=195},name=a par acomp times c par comp
    [,times
      [,for tree={grow''=-90}, name=a par acomp
      [,par,s sep=15mm [,name=a,label=-90:$\alpha$]  [,name=acomp,label=-90:$\llneg{\alpha}$] ]
      ]
      [,for tree={grow''=120}, name=c par ccomp
      [,par,s sep=15mm [,name=c,label=$\gamma$] [,name=ccomp,label=$\llneg{\gamma}$] ]
      ]
    ]
    ]
    [,for tree={grow''=60}, name=b par bcomp
    [,par,s sep=15mm [,name=b,label=$\beta$] [,name=bcomp,label=$\llneg{\beta}$]]]
    ]
    ]
    \draw[Blink] (a) to (acomp) ;
    \draw[Blink] (b) to (bcomp) ;
    \draw[Blink] (c) to (ccomp) ;
\end{forest}%
}
 \caption{Folding a dicograph proof net  into an \sp\ proof net  step by step ($\pi_1,\pi_2,\pi_3,\pi_4$) --- the conclusions are the black vertices.} 
\label{foldingStepByStep} 
\end{figure} 
%\caption{Progressively turning a dicograph proof net $\pi_1$ into a pomset proof net $\pi_4$ with one conclusion via some intermediate proof nets  $\pi_2$ and $\pi_3$. The conclusions of the pomset proof net and of the intermediate proof nets are emphasised by filled black dots.}
%\label{foldingStepByStep} 

Let $\pi=(B,D[F_1,\ldots,F_p])$ with $D$  a dicograph on the formulas $F_1,\ldots,F_n$ be an intermediate proof structure . 
A folding of $\pi$ is a simply a folding of $D[F_1,\ldots,F_p]$ as defined in Subsection \ref{foldunfoldsequent} (two equivalent formulas $F_i \widehat{*} F_j$ are replaced in $D$ by one formula $F_i * F_j$). An unfolding of  $\pi$ is simply an unfolding of  $D[F_1,\ldots,F_p]$ as defined in Subsection \ref{foldunfoldsequent} (a formula  $F_i * F_j$  is replaced by two equivalent formulas $F_i \widehat{*} F_j$).

\begin{proposition}\label{foldunfoldcorrect} 
Let $\pi_f$ and $\pi_u$ be two intermediate proof structures, with $\pi_u$ being an unfolding of $\pi_f$ --- or $\pi_f$ being a folding of $\pi_u$. 
The two following properties are equivalent: 
\begin{itemize} 
\item $\pi_u$ is correct. 
\item $\pi_f$ is correct. 
\end{itemize} 
\end{proposition} 

\begin{proof} 
This proof consists in a thorough examination of new \ael\ circuits that may appear during the transformation and of the edges that are chords and that may vanish during the transformation. \cite{Ret99romarr}. 
\end{proof} 

We now can give again the definition of pomset proof nets with link that appear in \cite{Ret93,Ret97tlca}: 

\begin{definition} 
A pomset proof structure with links is defined as a combination of links: every conclusion of a link is the premisse of at most one link, and every premisse of a link is the conclusion of exactly one link. FOrmulas that are not the premisse of any link are called conclusions of the proof net. Conclusions are conected with an $R$ \sp-order. 
\end{definition} 

Because of the shape of the RnB links the Criterion \ref{criterionHandsome}, \emph{every \ae-circuit contains a chord} becomes simpler: 
\begin{criterion} 
A pomset proof structure with links is correct whenever there is no \ael\ circuit.  
\end{criterion} 

A possible variant for pomset proof structures with links, consist in replacing axioms, that are a single B edge (cf. the links in Figure~\ref{rnblinks}) with a sequence of 
a B edge, and R symmetric edge and a B edge: 
\begin{center}
  \tikzset{vertex/.append style={inner sep=1.5pt}}
  \begin{tikzpicture}\pgfkeys{vertex/.append style={inner sep=1pt}}
  \node[vertex,label=180:$a$] (a) {} ;
  \node[vertex,label=0:$\llneg{a}$,right=3] (a comp) {} ;
  \BRBaxiom[3.5mm]{a}{a comp}{additional R edge}
\end{tikzpicture}
\tikzset{vertex/.append style={inner sep=2pt}}
\end{center}

This R edge which is not incident with any other R edge 
does not changes anything to the \ae\  
paths and circuits, nor to the correctness of the proof structure (the non trivial sense is proved in Proposition \ref{removing1}), but that way any B edge corresponds to a formula (while in handsome proof net, every vertex corresponds to an atom).

\subsection{Gammars with partial proof nets} 
\label{PNgrammar} 

In the nineties, Lecomte was aiming at extensions of the  Lambek grammars that would  handle relatively free word orders, discontinuous constituents 
and other tricky linguistic phenomena, but still within a logical framework --- as opposed to CCG which extends AB grammars with \emph{ad hoc} rewriting rules whose logical content is unclear. Grammars defined within a  logical framework have at least two advantages: rules remain general and the connection with semantics, logical formulas and lambda terms is a priori more transparent. Following a suggestion by Jean-Yves Girard,  Alain Lecomte contacted me 
just after I passed my PhD on pomset logic, so we proposed a kind of grammar with pomset logic. 
We  explored such a possiblity in \cite{LR95,LR96tar,LeRe96,LeRe97} and it was later improved by Sylvain  Pogodalla in \cite{Pog98}  (see also \cite{RetoreHDR}).  

We followed two guidelines: 
\begin{itemize}
    \item Words are associated not with formulas but with partial proof nets with a tree-like structure, in particular  they have a single output;
    \item word order is a partial order,  an \sp\ order described by the occurrences of the  $\bef$ connective in the proof net.
\end{itemize}  

An  analysis or \emph{parse structure} is a combination of the partial proof nets into a complete proof net with output $S$. The two ways to combine partial proof nets are by ``plugging" an hypothesis to the conclusion of another partial proof net, and to perform cuts between partial proof nets. 

Given that words label axioms, instead of having a single $B$ edge from $a \edge a^\perp$ 
we write a sequence of three edges, a $B$ edge, an $R$ edge, a $B$ edge, the middle one being labelled with the word
\begin{center}
  \tikzset{vertex/.append style={inner sep=1.5pt}}
  \begin{tikzpicture}\pgfkeys{vertex/.append style={inner sep=1pt}}
  \node[vertex,label=180:$a$] (a) {} ;
  \node[vertex,label=0:$\llneg{a}$,right=2] (a comp) {} ;
  \BRBaxiom[2.5mm]{a}{a comp}{word}
\end{tikzpicture}
\tikzset{vertex/.append style={inner sep=2pt}}
\end{center}
% $$a {\genfrac {}{}{2pt}{1}\qquad B}\! {\circ }\!{\genfrac {}{}{1pt}{1} {\ word\ } R}\! {\circ}\!{\genfrac {}{}{2pt}{1} \qquad B} a^\perp$$
this little variant changes nothing regarding the correctness of the proof net in terms of \ae\ paths.  

Rather than lengthy explanations, let us give two examples of a grammatical derivation in this framework. One may notice in the examples that the partial pomset proof nets that we use in the lexicon are of a restricted form: 

\begin{itemize} 
\item there are just two conclusions:
\begin{itemize} 
\item the output $b$  which is the syntactic category of the resulting phrase once  the required "arguments" have been provided; 
\item a conclusion $a^\perp \pa (X_1 \ts Y_1)\pa\cdots\pa(X_n \ts Y_n)$ without any $\ts$ connective in the $X_i$; 
\end{itemize} 
\item an axiom  connects $a^\perp$  in the conclusion with an $a$ in one of the $X_i$ --- with the corresponding word is the label of $a$; 
\end{itemize} 

In a first version we defined from the proof net an order between atoms (hence words) by ``there exists a directed path" from $a$ to $b$. However it is more convenient, in particular from a computational point of view, to label the proof net with \sp\ orders of words. Doing so is a computational improvement but those labels are fully determined by the proof net, they contain no  additional information. Here are the labelling rules: 
\begin{itemize} 
\item Initialisation: 
\begin{itemize} 
\item $a^\perp$ is labelled with the one point \sp\ order  consisting of the corresponding word; 
\item $X_i\ts Y_i$ is labelled with an empty \sp\ order.  
\end{itemize} 
\item Propagation: 
\begin{itemize} 
\item The two conclusions of a given axiom have the same label; 
\item One of the two premisses of a tensor link is labelled with the \sp\ order $R\pasp S$ the other by $R$ and the conclusion by $S$; 
\item The conclusion of a $\pa$ link is labelled $R\pasp S$ when the two premisses are labelled $R$ and $S$; 
\item The conclusion of a $\bef$ link is labelled $R\befsp S$ when the two premisses are labelled $R$ and $S$; \end{itemize} 
\end{itemize} 

The propagation rules always succeed because of the correctness criterion and of the tree like structure of the  partial proof nets. The propagation rules  yield a complete labelling of the proof net and the  \sp\ order that labels the output $S$ is the partial order over words. 

\begin{table}
    \centering
  \begin{tabular}{l@{\hspace{1cm}}l}
Pierre &
\begin{tikzpicture}[baseline=(current bounding box.center)]
  \node[vertex,label=180:$\llneg{\np}$] (npcomp pierre) {} ;
  \node[vertex,label=0:$\np$,right=2] (np pierre) {} ;
  \BRBaxiom[5mm]{npcomp pierre}{np pierre}{Pierre}
\end{tikzpicture}\\[0.8cm]
Marie &
\begin{tikzpicture}[baseline=(current bounding box.center)]
  \node[vertex,label=180:$\llneg{\np}$] (npcomp marie) {} ;
  \node[vertex,label=0:$\np$,right=2] (np marie) {} ;
  \BRBaxiom[5mm]{npcomp marie}{np marie}{Marie}
\end{tikzpicture}\\[0.8cm]
chanter&
\begin{tikzpicture}[baseline=(current bounding box.center)]
  \node[vertex,label=180:$\vinf$] (vinf chanter) {} ;
  \node[vertex,label=0:$\llneg{\vinf}$,right=2cm] (vinfcomp chanter) {} ;
  \BRBaxiom[5mm]{vinfcomp chanter}{vinf chanter}{chanter}
\end{tikzpicture}\\[1.5cm]
entend &
%\includegraphics{Pictures/entend.pdf}
%\hspace{-10.5cm} %%%%%%%%%%% PBHSPACE 
         \begin{forest} proof structure, with phantom node, commutative
  [,phantom
  [,label=0:$\llneg{v} \llpar ((\np \llprecedes v  \llprecedes (\np \llts v)) \llts \llneg{S})$
  [,par
  [,label=180:$\llneg{v}$,name=v comp]
  [,fit=rectangle
  [,times
  [,
  [,precedes
  [,label=180:$\np$,name=npsubj,baseline [] ]
  [,
  [,precedes
  [,label=180:$v$,name=v ]
  [,
  [,times
  [,label=180:$\np$,name=npobj [] ]
  [,label=135:$\vinf$,name=vinf [] ]
  ]
  ]
  ]
  ]
  ]
  ]
  [,label=-45:$\llneg{S}$,name=s comp]
  ]
  ]
  ]
  ]
  [,label=0:$S$,name=s,l*=3]
  ]
  \BRBaxiom[9.65cm]{v comp}{v}{entend}
  \BRBaxiom[5.75cm]{s comp}{s}{}
\end{forest}\\
\end{tabular}
\caption{A lexicon with partial pomset proof nets} 
\label{lexiquePEMC} 
\end{table}

\begin{figure}
    \centering
%Pierre 
%\includegraphics{Pictures/PierreEntendMarieChanter.pdf}
\begin{forest} proof structure, commutative, with phantom node
  [,phantom,fit=band,s sep=1.5cm
  [,label=0:$\llneg{\np}$,name=npobj comp,fit=band,l*=11]
  [,label=0:$\llneg{\np}$,name=npsubj comp,fit=band,l*=7]
  [,fit=band,tikz={\node [draw,black,inner sep=6mm,xshift=.8cm,dotted,ellipse,fit to =tree,label=above:entendre]  {};},label=-90:$\llneg{v} \llpar ((\np \llprecedes v  \llprecedes (\np \llts v)) \llts \llneg{S})$
  [,par
  [,label=180:$\llneg{v}$,name=v comp]
  [,fit=rectangle
  [,times
  [,
  [,precedes
  [,label=180:$\np$,name=npsubj]
  [,
  [,precedes
  [,label=180:$v$,name=v]
  [,
  [,times
  [,label=180:$\np$,name=npobj]
  [,label=135:$\vinf$,name=vinf]
  ]
  ]
  ]
  ]
  ]
  ]
  [,label=-45:$\llneg{S}$,name=s comp]
  ]
  ]
  ]
  ]
  [,fit=band,label=0:$S$,name=s,l*=3]
  [,label=0:$\llneg{\vinf}$,name=vinf comp,l*=11]
  ]
  \BRBaxiom[9.65]{v comp}{v}{entend}
  \BRBaxiom[5.75]{s comp}{s}{}
  \BRBaxiom[8.5][pos=.16]{npsubj comp}{npsubj}{Pierre}
  \BRBaxiom[11.5]{npobj comp}{npobj}{Marie}
  \BRBaxiom[11.5][pos=.35]{vinf comp}{vinf}{chanter}
\end{forest}
\caption{Analysis of a relatively free word order sentence --- order $Pierre\bef entend \bef (Marie \pa chanter)$} 
\label{PierreEntendMArieChanter} 
\end{figure} 

We give an example of a lexicon of an analysis of a relatively free word order phenomenon in French --- the lexicon is in Table~\ref{lexiquePEMC} and the analysis in Figure  \ref{PierreEntendMArieChanter}. One can say both "Pierre entend Marie chanter" (Pierre hears Mary singing) and "Pierre entend chanter Marie" (Pierre hears singing Mary).  Indeed when there is no object French accepts that the subject is after the verb, e.g. in relatives introduced by the relative pronoun "que/whom": ``Pierre que regarde Marie chante." (Pierre that Mary watches sings" and ``Pierre que Marie regarde chante." (Pierre that  Marie watches sings). Observe that there is a single analysis for the different possible word orders and not a different analysis for each word order.  

Using cuts, one is able, in addition to free word order phenomena,  to provide an account of discontinuous constituents, e.g. French negation ``ne \ldots pas". During cut elimination, the label splits into two parts so ``ne" and ``pas" go to their proper places, as shown in Figure \ref{NePas}. When cut is used, one may allow that incorrect lexical partial proof nets are associated with words, but after cut elimination the result, i.e. the linguistic analysis, ought to be a fully correct proof net; 

\begin{figure}
    \centering
\subfigure[
\begin{minipage}{\textwidth}The proof net made from the partial proof nets NE\ldots  PAS (discontinuous constituent) and from the partial proof net REGARDE, before cut-elimination.\end{minipage}]
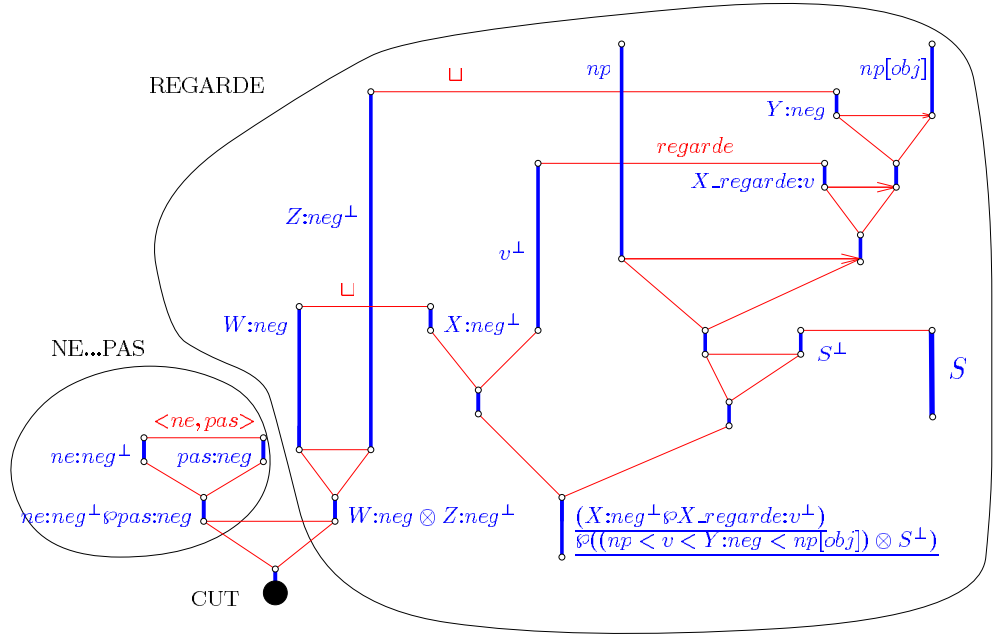
{%
  \begin{forest} proof structure, commutative,with phantom node,for tree ={l=15pt}
  [,phantom,s sep=1cm
  [,
  [,cut,s sep=1.3cm
  [,tikz={\node [draw,black,inner sep=1cm,xshift=-1cm,ellipse,dotted,fit to=tree,label={105:ne\ldots pas}]  {};},label={[name=lab]180:$\NE:\llneg{\NEG} \llpar \pas:\NEG$}
  [,par
  [,label=180:$\NE:\llneg{\NEG}$,name=ne neg comp] [,label=-85:$\pas:\NEG$,name=pas neg]
  ]]
  [,
  [,times,label=0:$W:\NEG\llts Z:\llneg{\NEG}$
  [,label=175:$W:\NEG$,name=w neg] [,label={-85:$Z:\llneg{\NEG}$},name=z neg comp]
  ]]
  ]
  ]
  [,tikz={\node [draw,black,inner sep=.7cm,xshift=-5mm,ellipse,dotted,fit to=tree,label={regarde}]  {};},label=-90:{$(X:\llneg{\NEG} \llpar X\_{\mathit{regarde}}:\llneg{v}) \llpar ((\np \llprecedes v \llprecedes Y:\NEG \llprecedes \np[\obj])\llts \llneg{S})$},fit=band
  [,par,s sep=2.7cm
  [[,par,s sep=1.5cm
  [,label=45:$X:\llneg{\NEG}$,name=x neg comp] [,label=0:$X\_{\mathit{regarde}}:\llneg{v}$,name=x v comp]
  ]]
  [[,times,s sep=1cm
  [[,precedes [,label=180:$\np$,name=np subj []][
  [,precedes [,label=180:$X\_{\mathit{regarde}}:v$,name=x v] [
  [,precedes [,label=165:$Y:\NEG$,name=y neg] [,label=0:{$\np [\obj]$},name=np obj []]
  ]
  ]
  ]
  ]]
  ]
  [,label=0:$\llneg{S}$,name=s comp]
  ]]
  ]]
  [,label=0:$S$,name=s,l*=5.6]
  ]
  \BRBaxiom[3.7]{ne neg comp}{pas neg}{$\langle \mathit{ne}, \mathit{pas} \rangle$}
  \BRBaxiom[3.7]{w neg}{x neg comp}{}
  \BRBaxiom[7.25]{z neg comp}{y neg}{}
  \BRBaxiom[6.3]{x v comp}{x v}{regarde}
  \BRBaxiom[3.5]{s comp}{s}{}
\end{forest}
}\\[-2ex] %% Pour ne pas que la légende couvre le numéro de page
\subfigure[
\begin{minipage}{\textwidth}The proof net analysing NE REGARDE PAS, after reduction, the three words are in the proper order.\end{minipage}]{%
\begin{forest} proof structure, commutative,with phantom node,for tree ={l=15pt}
  [,phantom,s sep=1.5cm
  [,label=-90:{$(\NE:\llneg{\NEG} \llpar \NE\_{\mathit{regarde}}:\llneg{v}) \llpar ((\np \llprecedes v \llprecedes \pas:\NEG \llprecedes \np[\obj])\llts \llneg{S})$},fit=band
  [,par,s sep=2.5cm
  [[,par
  [,label=180:$\NE:\llneg{\NEG}$,name=x neg comp] [,label=0:$\NE\_{\mathit{regarde}}:\llneg{v}$,name=x v comp]
  ]]
  [[,times
  [[,precedes [,label=180:$\np$,name=np subj []][
  [,precedes [,label=180:$\NE\_{\mathit{regarde}}:v$,name=x v] [
  [,precedes [,label=180:$Y:\NEG$,name=y neg] [,label=0:{$\np [\obj]$},name=np obj []]
  ]
  ]
  ]
  ]]
  ]
  [,label=0:$\llneg{S}$,name=s comp]
  ]]
  ]]
  [,label=0:$S$,name=s,l*=5.6]
  ]
  \BRBaxiom[7.2]{x neg comp}{y neg}{$\langle\NE,\pas\rangle$}
  \BRBaxiom[6.2]{x v comp}{x v}{regarde}
  \BRBaxiom[3.5]{s comp}{s}{}
\end{forest}
}
\caption{Handling discontinuous constituents in pomset proof nets}
\label{NePas} 
\end{figure}

It is difficult to say something on the generative capacity of  this grammatical formalism, because it produces (or recognises) \sp\ order of words and not chains of words --- and there are not so many such grammatical formalisms, en exception being \cite{LodayaWeil2001tcs}. 

\begin{theorem}[Pogodalla] 
Pomset grammars with a restricted form for partial pomset proof nets  yielding trees and total word orders is equivalent to Lexicalised Tree Adjoining Grammars \cite{Pog98}. 
\end{theorem} 

This is much more than languages 
that can be generated by Lambek grammars, that are context free. In both cases, parsing as proof search is NP complete -- trying all the possibilities in pomset grammar is in NP (and likely to be NP complete), and provability Lambek calculus has been shown to be NP complete \cite{Pen03} --- of course if the Lambek grammar is converted into an extremely large context free grammar using the result of Matti Pentus \cite{Pen93} parsing of Lambek grammars is polynomial, cubic or better in the number of words in the sentence. 

Especially when using cuts and tree-like partial proof net this calculus is close to several coding of LTAG in non commutative linear logic \`a la Lambek-Abrusci \cite{AFV96}.\footnote{The related work \cite{JK96} which also encodes TAGs in non commutative linear logic \`a la Lambek-Abrusci, presented with natural deduction, requires \emph{ad hoc} extensions of the non commutative linear logic like some crossing of the axioms which are excluded from those Lambek-Abruci logics \cite{Roo91,Ret96tal}.}

\section{Conclusion and perspective} 
\label{concl} 

We presented an overview of pomset logic with both published and unpublished results. Pomset logic is a variant of linear logic, as the Lambek calculus is, and it can be used for modelling grammar, in particular for natural language as the Lambek calculus can. Apart from this, as said in the introduction, Lambek calculus and pomset logic, are quite  different, although they are both non commutative variants of (multiplicative) linear logic. Lambek calculus is a non commutative restriction of intuitionistic linear logic, while pomset logic is a non commutative  extension of classical linear logic.

But perhaps the resemblance is more abstract than that. Indeed Lambek was surprised that with proof nets people intend to replace a syntactic calculus, an algebraic structure,  with graphical or geometrical objects. However for pomset logic, the best presentation is certainly the calculus of dicographs, which can be viewed as terms, and therefore belong to algebra. It is not surprising that Lambek preferred my algebraic correctness criterion that use coherence spaces \cite{Ret94rc,Ret97,RetoreHDR} (Theorem \ref{semcorrect} in this paper), to the double trip condition of Girard of \cite{Gir87}. 

This presentation is by no means the necrology of pomset logic. For instance, Slavnov recently proposed a  sequent calculus which is complete w.r.t. pomset proof nets \cite{Slavnov2019}.  In his sequent calculus, multisets of formulas are endowed with binary relations on sequences of $n$ conclusions, and ``before" $\bef$ is, with respect to Slavnov's calulus,  a collapse of two connectives namely a before $\bef$ which behaves like  a 
``times" $\ts$ and a before $\bef$ which behaves  like a ``par" $\pa$.  Stra{\ss}burger contributed to pomset logic with the counterexample in the paper,  but his major contributions to pomset logic   are the development of Deep Inference (including SBV that we discussed in this paper)  and the comparison between pomset logic defined with handsome proof nets,   Deep Inference \cite{GugStr01}, proofs without syntax \cite{HHS2019lics};  Stra{\ss}burger is  presently exploring new ideas on applications of  pomset logic  to safety and  privacy together with Horne.  Guglielmi in e.g. \cite{gug2017lyon} is tuning the syntactic  rules of a self dual modality $\flag$ which is to  $\bef$ what $?$ and $!$ are to $\otimes$ and $\parr$: $\flag A$ is linearly 
equivalent to $\flag A\bef \flag A$ and $A$ is a retract of $\flag A$
that we defined in coherence spaces  \cite{Ret94mod,Retore2020tlla}.  Nguyen is exploring the combinatorial properties of handsome proof nets for pomset logic \cite{nguyen2019proof}. All those topical research directions related to pomset logic are excellent reasons to keep on exploring the pomset logic approach to non commutativity in logic, which is, up to now, the unique alternative to the kind of non commutativity in logic introduced in 1958 by  Joachim Lambek (1922-2014).

\paragraph{Acknowledgement} \emph{
I would like to thank Claudia Casadio and Philip Scott who have been patient enough for this article to be included in the volume. I would also like to thank the reviewers and Philip Scott again because they helped me a lot to improve the first draft of this paper, both in its form and contents. Let me also thank Matteo Acclavio, Tito Nguyen, Sylvain Pogodalla and especially Lutz Stra{\ss}Burger for helpful discussions on handsome proof nets, pomset logic and deep inference. This chapter owes a lot to  Sylvain Pogodalla: without his help, no pictures, no  proofnets, no chapter!}

\bibliographystyle{acm} 

\bibliography{bigbiblio} 

\begin{thebibliography}{10}

\bibitem{Abru91}
{\sc Abrusci, V.~M.}
\newblock Phase semantics and sequent calculus for pure non-commutative
  classical linear logic.
\newblock {\em Journal of Symbolic Logic 56}, 4 (1991), 1403--1451.

\bibitem{AFV96}
{\sc Abrusci, V.~M., Fouquer\'e, C., and Vauzeilles, J.}
\newblock Tree adjoining grammar and non-commutative linear logic.
\newblock In Retor\'e \cite{LACL96}, pp.~13--17.

\bibitem{AR99}
{\sc Abrusci, V.~M., and Ruet, P.}
\newblock Non-commutative logic {I}: the multiplicative fragment.
\newblock {\em Annals of Pure and Applied Logic\/} (1999).

\bibitem{BGR97}
{\sc Bechet, D., de~Groote, P., and Retor\'e, C.}
\newblock A complete axiomatisation of the inclusion of series-parallel partial
  orders.
\newblock In {\em Rewriting Techniques and Applications, RTA`97\/} (1997),
  H.~Comon, Ed., vol.~1232 of {\em LNCS}, Springer Verlag, pp.~230--240.

\bibitem{dGLam2002}
{\sc de~Groote, P., and Lamarche, F.}
\newblock Classical non-associative lambek calculus.
\newblock {\em Studia Logica 71}, 3 (2002), 355--388.

\bibitem{FR94}
{\sc Fleury, A., and Retor\'{e}, C.}
\newblock The mix rule.
\newblock {\em Mathematical Structures in Computer Science 4}, 2 (1994),
  273--285.

\bibitem{Gir86}
{\sc Girard, J.-Y.}
\newblock The system {F} of variable types: Fifteen years later.
\newblock {\em Theoretical Computer Science 45\/} (1986), 159--192.

\bibitem{Gir87}
{\sc Girard, J.-Y.}
\newblock Linear logic.
\newblock {\em Theoretical Computer Science 50}, 1 (1987), 1--102.

\bibitem{GLT88}
{\sc Girard, J.-Y., Lafont, Y., and Taylor, P.}
\newblock {\em Proofs and Types}.
\newblock No.~7 in Cambridge Tracts in Theoretical Computer Science. Cambridge
  University Press, 1988.

\bibitem{Gug94}
{\sc Guglielmi, A.}
\newblock Concurrency and plan generation in a logic programming language with
  a sequential operator.
\newblock In {\em International Conference on Logic Programming\/} (Genova,
  1994), P.~van Hentenryck, Ed., M.I.T. Press, pp.~240--254.

\bibitem{Gug99}
{\sc Guglielmi, A.}
\newblock A calculus of order and interaction.
\newblock Tech. Rep. WV-99-04, Dresden University of Technology, 1999.

\bibitem{Gug2007}
{\sc Guglielmi, A.}
\newblock A system of interaction and structure.
\newblock {\em ACM Transactions on Computational Logic 8}, 1 (2007), 1--64.

\bibitem{gug2017lyon}
{\sc Guglielmi, A.}
\newblock Decoupling normalization mechanisms with an eye toward concurrency.
\newblock Talk at ENS LYON seminar: programs and proofs., Aril 2017.
\newblock Slides: \url{http://cs.bath.ac.uk/ag/t/DNMWAETC.pdf}.

\bibitem{GugStr01}
{\sc Guglielmi, A., and Stra{\ss}burger, L.}
\newblock Non-commutativity and {MELL} in the calculus of structures.
\newblock In {\em CSL 2001\/} (2001), L.~Fribourg, Ed., vol.~2142 of {\em
  Lecture Notes in Computer Science}, Springer-Verlag, pp.~54--68.

\bibitem{HHS2019lics}
{\sc Heijltjes, W.~B., Hughes, D. J.~D., and Stra{\ss}burger, L.}
\newblock Intuitionistic proofs without syntax.
\newblock In {\em 34th Annual {ACM/IEEE} Symposium on Logic in Computer
  Science, {LICS} 2019, Vancouver, BC, Canada, June 24-27, 2019\/} (2019),
  {IEEE}, pp.~1--13.

\bibitem{JK96}
{\sc Joshi, A., and Kulick, S.}
\newblock Partial proof trees, resource sensitive logics and syntactic
  constraints.
\newblock In Retor\'e \cite{LACL96}, pp.~21--42.

\bibitem{Kel85}
{\sc Kelly, D.}
\newblock Comparability graphs.
\newblock In Rival \cite{Riv85}.
\newblock pp 3--40.

\bibitem{LR96}
{\sc Lamarche, F., and Retor\'e, C.}
\newblock Proof nets for the {L}ambek calculus -- an overview.
\newblock In {\em Third Roma Workshop: Proofs and Linguistics Categories --
  Applications of Logic to the analysis and implementation of Natural
  Language\/} (Bologna, 1996), V.~M. Abrusci and C.~Casadio, Eds., CLUEB,
  pp.~241--262.

\bibitem{Lam58}
{\sc {L}ambek, J.}
\newblock The mathematics of sentence structure.
\newblock {\em American mathematical monthly\/} (1958), 154--170.

\bibitem{LS86}
{\sc Lambek, J., and Scott, P.~J.}
\newblock {\em Introduction to Higher Order Categorical Logic}, vol.~7 of {\em
  Cambridge Studies in Advanced Mathematics}.
\newblock Cambridge University Press, 1986.

\bibitem{LR95}
{\sc Lecomte, A., and Retor{\'e}, C.}
\newblock Pomset logic as an alternative categorial grammar.
\newblock In {\em Formal Grammar\/} (Barcelona, 1995), G.~Morrill and
  R.~Oehrle, Eds., FoLLI, pp.~181--196.

\bibitem{LeRe96}
{\sc Lecomte, A., and Retor\'e, C.}
\newblock Words as modules and modules as partial proof nets in a lexicalised
  grammar.
\newblock In {\em Third Roma Workshop: Proofs and Linguistics Categories --
  Applications of Logic to the analysis and implementation of Natural
  Language\/} (1996), V.~M. Abrusci and C.~Casadio, Eds., Bologna:CLUEB,
  pp.~187--198.

\bibitem{LeRe97}
{\sc Lecomte, A., and Retor\'e, C.}
\newblock Logique des ressources et r{\'e}seaux syntaxiques.
\newblock In {\em 4$^e$ conf\'erence sur le Traitement automatique du langage
  naturel, TALN'97\/} (Grenoble, 1997), D.~Genthial, Ed., pp.~70--83.

\bibitem{LR96tar}
{\sc Lecomte, A., and Retor\'e, C.}
\newblock Words as modules: a lexicalised grammar in the framework of linear
  logic proof nets.
\newblock In {\em Mathematical and Computational Analysis of Natural Language
  --- selected papers from ICML`96\/} (1998), C.~Martin-Vide, Ed., vol.~45 of
  {\em Studies in Functional and Structural Linguistics}, John Benjamins
  publishing company, pp.~129--144.

\bibitem{Loa94lics}
{\sc Loader, R.}
\newblock Linear logic, totality and full completeness.
\newblock In {\em LICS: IEEE Symposium on Logic in Computer Science\/} (1994).

\bibitem{LodayaWeil2001tcs}
{\sc Lodaya, K., and Weil, P.}
\newblock Series-parallel languages and the bounded-width property.
\newblock {\em Theor. Comput. Sci. 237}, 1-2 (2000), 347--380.

\bibitem{Moh85}
{\sc M{\"o}hring, R.~H.}
\newblock Algorithmic aspects of comparability graphs.
\newblock In Rival \cite{Riv85}, pp.~41--101.

\bibitem{nguyen2019proof}
{\sc Nguy{\^e}n, L. T.~D.}
\newblock Proof nets through the lens of graph theory: a compilation of
  remarks, 2019.

\bibitem{Pen93}
{\sc Pentus, M.}
\newblock {L}ambek grammars are context-free.
\newblock In {\em Logic in Computer Science\/} (1993), IEEE Computer Society
  Press.

\bibitem{Pen03}
{\sc Pentus, M.}
\newblock {L}ambek calculus is {NP}-complete.
\newblock Tech. Rep. TR-2003005, CUNY - City University of New York, 2003.
\newblock http://www.cs.gc.cuny.edu/tr/.

\bibitem{Pog98}
{\sc Pogodalla, S.}
\newblock Lexicalized proof-nets and {TAG}s.
\newblock In {\em Logical Aspects of Computational Linguistics, LACL`98,
  selected papers\/} (2001), M.~Moortgat, Ed., no.~2014 in LNCS/LNAI,
  Springer-Verlag, pp.~230--250.

\bibitem{PogoReto04}
{\sc Pogodalla, S., and Retor{\'e}, C.}
\newblock Handsome non-commutative proof-nets: perfect matchings,
  series-parallel orders and hamiltonian circuits.
\newblock In {\em Categorial grammars -- an efficient tool for natural language
  processing\/} (2004).

\bibitem{Red93}
{\sc Reddy, U.~S.}
\newblock A linear logic model of state.
\newblock Electronic manuscript, University of Illinois (anonymous FTP from
  cs.uiuc.edu), 1993.

\bibitem{Red96}
{\sc Reddy, U.~S.}
\newblock Global state considered unnecessary: An introduction to object-based
  semantics.
\newblock {\em Journal of Lisp and Symbolic Computation\/} (1996).
\newblock (to appear.).

\bibitem{Ret93b}
{\sc Retor\'e, C.}
\newblock Graph theory from linear logic: aggregates.
\newblock Pr\'epublication~47, Equipe de Logique, Universit\'e Paris 7, 1993.

\bibitem{Ret93}
{\sc Retor\'e, C.}
\newblock {\em R\'eseaux et S\'equents Ordonn\'es}.
\newblock Th\`ese de {D}octorat, sp\'ecialit\'e {M}ath\'ematiques, Universit\'e
  Paris 7, f\'evrier 1993.

\bibitem{Ret94rc}
{\sc Retor\'e, C.}
\newblock On the relation between coherence semantics and multiplicative proof
  nets.
\newblock Rapport de {R}echerche RR-2430, INRIA, d\'ecembre 1994.

\bibitem{Ret94mod}
{\sc Retor\'e, C.}
\newblock A self-dual modality for "before" in the category of coherence spaces
  and in the category of hypercoherences.
\newblock Rapport de {R}echerche RR-2432, INRIA, d\'ecembre 1994.

\bibitem{Ret96tal}
{\sc Retor\'e, C.}
\newblock Calcul de {L}ambek et logique lin\'eaire.
\newblock {\em Traitement Automatique des Langues 37}, 2 (1996), 39--70.

\bibitem{Ret96entcs}
{\sc Retor\'{e}, C.}
\newblock Perfect matchings and series-parallel graphs: multiplicative proof
  nets as {R}\&{B}-graphs.
\newblock In {\em Linear`96\/} (1996), J.-Y. Girard, M.~Okada, and A.~Scedrov,
  Eds., vol.~3 of {\em Electronic Notes in Theoretical Science}, Elsevier.
\newblock http://www.elsevier.nl/.

\bibitem{LACL96}
{\sc Retor\'e, C.}, Ed.
\newblock {\em Logical Aspects of Computational Linguistics, LACL`96\/} (1997),
  vol.~1328 of {\em LNCS/LNAI}, Springer-Verlag.

\bibitem{Ret97tlca}
{\sc Retor\'e, C.}
\newblock Pomset logic: a non-commutative extension of classical linear logic.
\newblock In {\em Typed Lambda Calculus and Applications, TLCA'97\/} (1997),
  P.~de~Groote and J.~R. Hindley, Eds., vol.~1210 of {\em LNCS}, pp.~300--318.

\bibitem{Ret97}
{\sc Retor\'e, C.}
\newblock A semantic characterisation of the correctness of a proof net.
\newblock {\em Mathematical Structures in Computer Science 7}, 5 (1997),
  445--452.

\bibitem{Ret99rr}
{\sc Retor{\'{e}}, C.}
\newblock Handsome proof-nets: {R}\&{B}-graphs, perfect matchings and
  series-parallel graphs.
\newblock Rapport de {R}echerche RR-3652, INRIA, March 1999.

\bibitem{Ret98roma}
{\sc Retor\'e, C.}
\newblock Pomset logic as a calculus of directed cographs.
\newblock In {\em Dynamic Perspectives in Logic and Linguistics: Proof
  Theoretical Dimensions of Communication Processes,Proceedings of the 4th Roma
  Workshop}, V.~M. Abrusci and C.~Casadio, Eds. Bulzoni, Roma, 1999,
  pp.~221--247.
\newblock See also \cite{Ret99romarr}.

\bibitem{Ret99romarr}
{\sc Retor{\'{e}}, C.}
\newblock Pomset logic as a calculus of directed cographs.
\newblock Rapport de {R}echerche RR-3714, INRIA, March 1999.
\newblock See also \cite{Ret98roma}.

\bibitem{RetoreHDR}
{\sc Retor\'e, C.}
\newblock {\em Logique lin{\'e}aire et syntaxe des langues}.
\newblock M{\'e}moire d'{H}abilitation {\`a} {D}iriger des {R}echerches,
  Universit{\'e} de Nantes, Janvier 2002.
\newblock \url{https://tel.archives-ouvertes.fr/tel-00354041}.

\bibitem{Ret03tcs}
{\sc Retor\'e, C.}
\newblock Handsome proof-nets: perfect matchings and cographs.
\newblock {\em Theoretical Computer Science 294}, 3 (2003), 473--488.
\newblock See also the more thorough report \cite{Ret99rr}.

\bibitem{Retore2020tlla}
{\sc Retor{\'e}, C.}
\newblock A self-dual modality for non commutative contraction and duplication
  in the category of coherence space.
\newblock In {\em Linearity \& Trends in Linear Logic and Applications\/}
  (2020), T.~Ehrhard, U.~dal Lago, and V.~de~Paiva, Eds.

\bibitem{Riv85}
{\sc Rival, I.}, Ed.
\newblock {\em Graphs and Order}, vol.~147 of {\em NATO ASI series C}.
\newblock Kluwer, 1985.

\bibitem{Roo91}
{\sc Roorda, D.}
\newblock {\em Resource logic: proof theoretical investigations}.
\newblock PhD thesis, FWI, Universiteit van Amsterdam, 1991.

\bibitem{Roo92}
{\sc Roorda, D.}
\newblock Proof nets for {L}ambek calculus.
\newblock {\em Logic and Computation 2}, 2 (1992), 211--233.

\bibitem{Ruet97}
{\sc Ruet, P.}
\newblock {\em Logique non-commutative et programmation concurrente}.
\newblock Th{\`e}se de doctorat, sp{\'e}cialit{\'e} logique et fondements de
  l'informatique, Universit{\'e} Paris 7, 1997.

\bibitem{Slavnov2019}
{\sc Slavnov, S.}
\newblock {On noncommutative extensions of linear logic}.
\newblock {\em {Logical Methods in Computer Science} {Volume 15, Issue 3}\/}
  (Sept. 2019).

\bibitem{Strassburger03phd}
{\sc Stra{\ss}burger, L.}
\newblock {\em Linear Logic and Noncommutativity in the Calculus of
  Structures}.
\newblock PhD thesis, Technische Universit\"at Dresden, 2003.

\bibitem{Yet90}
{\sc Yetter, D.~N.}
\newblock Quantales and (non-commutative) linear logic.
\newblock {\em Journal of Symbolic Logic 55\/} (1990), 41--64.

\end{thebibliography}

%\clearpage
%\thispagestyle{empty}
%\hfill
%\clearpage

\end{document}